%% file: 38_cyclic_LF_inf.tex
\documentclass[runningheads]{llncs}
\pdfoutput=1
\input{macros}
\newtoggle{journalversion}
\togglefalse{journalversion}

\iftoggle{journalversion}{
}{
    \pagestyle{plain}
}

\iftoggle{journalversion}{
    \includeonly{}
}{
    \includeonly{appendix}
}

\newcommand{\oftheextendedversion}[1][]{\iftoggle{journalversion}{ of the extended version{#1}}{}}

\author{Zhibo Chen\,\orcidlink{0000-0003-0045-5024} \and
Frank Pfenning\,\orcidlink{0000-0002-8279-5817}
\\
\email{zhiboc@andrew.cmu.edu}
\email{fp@cs.cmu.edu}
}
\authorrunning{Z. Chen and F. Pfenning}
\institute{
Carnegie Mellon University, Pittsburgh, PA, USA
}

\title{A Logical Framework with \\ Higher-Order Rational (Circular) Terms}

\begin{document}

\maketitle
\begin{abstract}
    Logical frameworks provide natural and direct ways of specifying and
    reasoning within deductive systems.
    The logical framework LF and subsequent developments focus on finitary proof systems, making
    the formalization of circular proof systems in such logical frameworks a cumbersome and awkward task.
    To address this issue, we propose $\CoLFCircular$, a 
    conservative extension of LF with higher-order rational terms
    and mixed inductive and coinductive definitions.
    In this framework, two terms are equal if they unfold to the same infinite regular B\"ohm tree.
    Both term equality and type checking are decidable in $\CoLFCircular$.
    We illustrate the elegance and expressive power of the framework with
    several small case studies.
\end{abstract}

\begin{keywords}
Logical Frameworks, Circular Proofs, Regular B\"ohm Trees
\end{keywords}
    
\section{Introduction}

A logical framework provides a uniform way of formalizing and mechanically checking derivations for
a variety of deductive systems common in the definitions of logics and programming languages.  In
this paper we propose a conservative extension of the logical framework LF~\cite{Harper93jacm} to
support direct representations of rational (circular) terms and deductions.

The main methodology of a logical framework is to establish a bijective correspondence between
derivations of a judgment in the object logic and canonical terms of a type in the framework.  In
this way, proof checking in the object logic is reduced to type checking in the framework.  One
notable feature of LF is the use of abstract binding trees, where substitution in the object logic
can be encoded as substitution in the framework, leading to elegant encodings.  On the other hand,
encodings of rational terms, circular derivations, and their equality relations are rather
cumbersome.  We therefore propose the logical framework $\CoLFCircular$ as a conservative extension
of LF in which both circular syntactic objects and derivations in an object logic can be elegantly
represented as higher-order rational dependently typed terms.  This makes $\CoLFCircular$ a uniform
framework for formalizing proof systems on cyclic structures.  We prove the decidability of type
checking and soundness of equality checking of higher-order rational terms.

While CoLF allows formalization of circular derivations, proofs by coinduction \emph{about} such
circular encodings can only be represented as relations in CoLF, mirroring a similar limitation of
LF regarding induction.  In future work, we plan to extend CoLF to support checking of
meta-theoretic properties of encodings analogous to the way Twelf~\cite{Pfenning98guide} can check
properties of encodings in LF.


The main contributions of this paper are:

\begin{itemize}
    \item The type theory of a logical framework with higher-order rational terms. The theory allows natural and adequate 
    representations of circular objects and circular derivations (Section~\ref{section:implementing_the_infinitary_theory_a_decidable_rational_fragment}).
    \item A decidable trace condition for ensuring the validity of circular terms and derivations arising from mixed inductive and coinductive definitions (Section~\ref{subsection:trace_condition}).
    \item A sound and complete algorithm to decide the equality of two higher-order rational terms (Section~\ref{subsection:term_equality}).
    \item A proof of decidability of type-checking in the framework (Section~\ref{subsection:metatheorems_after_type_checking_rules}).
    \item Case studies of encoding subtyping derivations of recursive types (Section~\ref{section:encoding_subtyping_systems_for_recursive_types}).
\end{itemize}

\iftoggle{journalversion}{ An extended version of this paper, available at \url{https://arxiv.org/abs/2210.06663}, has an appendix that
contains additional materials.}{}
We have implemented $\CoLFCircular$ in OCaml and the implementation can be accessed at 
\url{https://www.andrew.cmu.edu/user/zhiboc/colf.html}.
An additional case study of the meta-encoding the term model of CoLF in CoLF is presented in 
Appendix~\ref{appendix:encoding_cyclic_lambda_terms}\oftheextendedversion.

\section{Mixed Inductive and Coinductive Definitions}
\label{section:mixed_inductive_and_coinductive_definitions}
We motivate our design through simple examples of natural numbers, conatural numbers, 
and finitely padded streams. The examples serve to illustrate 
the idea of coinductive interpretations, and they do not involve dependent types or higher-order terms. 
More complex examples will be introduced later in the case studies (Section~\ref{section:encoding_subtyping_systems_for_recursive_types}). 

\begin{figure}[t]
\begin{tabular}{|l|l|}
\hline
\begin{minipage}{0.44\linewidth}
\vspace{.5ex}

\verb$nat : type.$

\verb$zero : nat.$

\verb$succ : nat -> nat.$
\vspace{1em}

\verb$w1 : nat = succ w1.$ \red{(not valid)}

\hrule
\vspace{0.5ex}

\verb$conat : cotype.$

\verb$cozero : conat.$

\verb$cosucc : conat -> conat.$

\vspace{1em}

\verb$w2 : conat = cosucc w2.$

\verb$w3 : conat = cosucc (cosucc w3).$

\vspace{1em}

\verb$eq : conat -> conat -> type.$

\verb$eq/refl : eq N N.$

\verb$eqw2w3 : eq w2 w3 = eq/refl.$
\end{minipage}
&
\begin{minipage}{0.52\linewidth}

\verb$padding : type.$

\verb$pstream : cotype.$

\verb$cocons : nat -> padding -> pstream.$

\verb$pad : padding -> padding.$

\verb$next : pstream -> padding.$

\vspace{1em}
\verb$s1 : pstream = cocons (succ zero) $

\verb$           (pad (pad (next s1))).$

\verb$p2 : padding = pad p2. $ \red{(not valid)}

\verb$s3 : pstream = cocons zero (next s3). $

\verb$s4 : pstream = cocons zero p5.$

\verb$p5 : padding = next s4.$

\verb$p6 : padding = pad p7.$ \red{(not valid)}

\verb$p7 : padding = pad p6.$ \red{(not valid)}

\end{minipage}
\\
\hline
\end{tabular}
\caption{Signatures and Examples for Section \ref{section:mixed_inductive_and_coinductive_definitions}}
\label{fig:sigantures_for_section_motivation}
\end{figure}
\subsubsection{Natural Numbers.}

The set of natural numbers is inductively generated by 
zero and successor. In a logical framework such as LF, one would encode
natural numbers as the signature consisting of the first three lines in the top left part of Fig.~\ref{fig:sigantures_for_section_motivation}.

The type theory ensures that canonical terms of the type $\mt{nat}$ are in one-to-one correspondence with 
the natural numbers. Specifically the \textit{infinite} stack of successors $\mt{succ\, (succ\, (succ\, \dots))}$ is not a 
valid term of type $\mt{nat}$. Therefore, the circular term \verb$w1$ 
is not a valid term.

\subsubsection{Conatural Numbers.}

We may naturally specify that a type admits a coinductive interpretation 
by introducing a new syntactic kind $\cotype$.
The kind $\cotype$ behaves just like the kind $\type$ except that 
now the terms under $\cotype$  are allowed to be circular.
 A slightly adapted signature would 
encode the set of conatural numbers, shown as the first three lines in the bottom left part of Fig.~\ref{fig:sigantures_for_section_motivation}.

Because $\mt{conat}$ is a coinductive type, the canonical forms of type 
$\mt{conat}$ includes $\mt{cosucc}^n \, \mt{cozero}$ for all $n$ and the infinite stack 
of $\mt{cosucc}$, which is in one to one correspondence with the set of conatural numbers.
Specifically, the infinite stack of \verb$cosucc$, may be represented by the valid circular term \verb$w2$ as in Fig.~\ref{fig:sigantures_for_section_motivation}.
The equality of terms in CoLF is the equality of the infinite trees generated by unfolding the 
terms, which corresponds to a bisimulation between circular terms. For example, an alternative representation of the infinite stack of \verb$cosucc$ is 
the term \verb$w3$, and CoLF will treat \verb$w2$ and \verb$w3$ as equal terms, as shown by the last three lines 
in the bottom left part of Fig.~\ref{fig:sigantures_for_section_motivation}. The terms \verb$w2$ and \verb$w3$ 
are proved equal by reflexivity. On the other hand, a formulation of conats in LF would involve an explicit constructor, e.g. \verb$mu : (conat -> conat) -> conat$.
The encoding of equality is now complicated and  
one needs to work with an explicit equality judgment whenever a \verb$conat$ is used.
Functions defined by coinduction (e.g., bisimulation in 
Appendix~\ref{appendix:bisimulation_example}\oftheextendedversion) need to be encoded as relations in CoLF.

\subsection{Finitely Padded Rational Streams}
\label{subsection:finitely_padded_streams}
As an example of mixed inductive and coinductive definition, we consider rational streams of natural numbers 
with finite paddings in between. These streams are special instances
of left-fair streams \cite{Basold18phd}. 
We define streams coinductively and define paddings inductively,
 such that there are
infinitely many numbers in the stream but only finitely 
many paddings between numbers, shown in the signature consisting of first five lines 
in the right column of Fig.~\ref{fig:sigantures_for_section_motivation}.
For example,
the term \verb$s1$ in Fig.~\ref{fig:sigantures_for_section_motivation} represents
a stream of natural number $1$'s with two paddings in between.
Because \verb$padding$ is a \verb$type$, the term \verb$p2$ is not valid, as it is essentially an infinite stack of 
\verb$pad$ constructors.
 Definitions in a CoLF signature can refer to each other. Thus,
 the terms \verb$s3$ and \verb$s4$ denote the same padded stream, and
 the  terms \verb$p6$, \verb$p7$ and \verb$p2$ denote the same invalid stream consisting of purely paddings.

\subsubsection{Priorities.}

 To ensure the adequacy of representation, types of kind $\cotype$ admit circular terms while 
types of kind $\type$ admit only finitary terms. It is obvious that the
circular term \verb$w1$ 
is \emph{not} a valid term of type \verb$nat$ due to the presence of 
an infinite stack of inductive constructors, 
and the circular term \verb$w2$ 
is a valid term of type \verb$conat$ because it is a stack of coinductive constructors. However, when we have 
both inductive and coinductive types, 
it is unclear whether a circular term (e.g. \verb$s1$) is valid.
Historically, priorities are used to resolve this ambiguity \cite{Charatonik98lics}.
A priority is assigned to each inductive or coinductive type, 
and constructors inherit priorities from their types.
Constructors with the highest priority types are then viewed as primary. 
In CoLF, priorities are determined by the order of their declarations.
Type families declared later have higher priorities than those declared earlier.
In this way,
 the type \verb$pstream$ has higher priority than the type \verb$padding$. 
 Constructor \verb$cocons$ inherits the priority of \verb$pstream$, and the term \verb$s1$ is viewed as an infinite stack of \verb$cocons$ and is thus valid. Similarly, 
terms \verb$s3$ and \verb$s4$ are also valid. If we switch the order of declaration of \verb$padding$ and \verb$pstream$ 
(thereby switching their priorities), then terms \verb$s1$, \verb$s3$, and \verb$s4$ are no longer valid.

\section{The Type Theory}
\label{section:implementing_the_infinitary_theory_a_decidable_rational_fragment}
We formulate the type theory of $\CoLFCircular$, a dependent type theory with higher-order rational terms and decidable type checking. 
The higher-order rational terms correspond to $\bot$-free regular B\"{o}hm trees \cite{Huet98mscs} and have decidable equality.

\subsection{Higher-Order Rational Terms}

When we consider first order terms (terms without $\lambda$-binders), 
 the rational terms are terms with only finitely many distinct subterms, and thus their equality is decidable.
We translate this intuition to the higher-order setting.
 The higher-order rational terms are those with finitely many subterms 
up to renaming of free and bound variables. 
We give several examples of rational and non-rational terms using the signatures in 
Section \ref{section:mixed_inductive_and_coinductive_definitions}. 
\begin{enumerate}
    \item The term \verb$w2$ in Fig.~\ref{fig:sigantures_for_section_motivation} is a first-order rational term.
\item
A stream counting up from zero $\mb{up}_0 = \m{cocons}\, \m{zero}\, (\m{next}\, (\m{cocons}\, (\m{succ}\, \m{zero})\,\\ (\m{next}\, (\dots))))$
is a first-order term that is not rational.
\item
A stream that repeats its argument  $\mb{R_2} = \lambda x.\, \m{cocons}\, x\, (\m{next}\, (\mb{R_2}\, x))$ 
is a higher-order rational term.
\item
A stream that counts up from a given number $\mb{up} = \lambda x.\, \m{cocons}\, x\, (\m{next}\, (\mb{up}\,\\ (\m{succ}\, x)))$
is \textit{not} a rational higher-order term.
\end{enumerate}

In the definitions above, bolded symbols on the left of the equality signs are called recursion constants. 
It is crucial that in higher-order rational terms, all arguments to  recursion constants are bound variables and not 
other kinds of terms.
 We call this restriction the \emph{prepattern restriction} as it is similar to Miller's pattern restriction \cite{Miller91jlc}
except that we allow repetition of arguments.
The prepattern restriction marks the key difference between the higher-order rational term $\mb{R_2}$ and the infinitary term $\mb{up}$.
The term $\mb{up}$ is not rational because the argument to $\mb{up}$ is $\m{succ}\, x$, which is not a bound variable.

\subsection{Syntax}
We build subsequent developments on canonical LF \cite{Harper07jfp}, a formulation of the LF type theory where terms are always in 
their canonical form. Canonical forms do not 
contain  
$\beta$-redexes and are fully $\eta$-expanded with respect to their typing, 
supporting bijective correspondences between object logic derivations and the terms of the framework. 
One drawback of this presentation is that canonical terms are not closed 
under syntactic substitutions, and the technique of hereditary substitution 
addresses this problem \cite{Watkins02tr}. 

The syntax of the theory follows the grammar shown in Fig. \ref{fig:syntax_for_colf_two_circular}.
We use the standard notion of spines. For example, a term $x\, M_1\, M_2\, M_3$ will be written 
as $x \cdot (M_1; M_2; M_3)$ where $x$ is the head and $M_1; M_2; M_3$ is the spine. 
To express rational terms, we add recursive definitions of the form $r : A = M$ to the signature, 
where $M$ must be contractive (judgment $M\iscontractive$) in that the head of $M$ must be a constant or a variable.
Recursive definitions look like notational definitions \cite{Pfenning98pstt}, but their semantics are very different. 
Recursive definitions 
are interpreted recursively in that the definition $M$ may mention the recursion constant $r$, and 
other recursion constants including those defined later in the signature, while
notational definitions in LF \cite{Pfenning98pstt} cannot be recursive.
Recursion constants are treated specially as a syntactic entity that is different 
from variables or constructors (nonrecursive constants). 
To ensure the conservativity over LF, we further require all definitions in $\Sigma$ to be linearly 
ordered. That is, only in the body of a recursive definition can we ``forward reference'', and we can only 
forward reference other recursion constants. All other declarations must strictly refer to names that have been defined previously.
We write $\lambda\overline x$ and $\overline M$ to mean a sequence of $\lambda$-abstractions 
and a sequence of terms respectively. We write $x, y, z$ for variables, $c, d$ for term constants (also called constructors), $a$ for type family constants, 
and $r, r', r''$ for recursion constants.

To enforce the prepattern restriction, we use a technical device called \emph{prepattern $\Pi$-abstractions}, 
and associated notion of \emph{prepattern variables} and \emph{prepattern spines}.
Prepattern $\Pi$-abstractions are written as $\Pi x \prepatof A_2.\, A_1$, and $x$ will 
be a prepattern variable (written $x\prepatof A_2$) in $A_1$. Moreover, in $A_1$, if $y$ is a variable  
of a prepattern type $\Pi w \prepatof A_2. B$, then the prepattern application of $y$ to $x$ will 
be realized as the head $y$ followed by a prepattern spine $([x])$, written $y\cdot ([x])$.
The semantics is that prepattern variables may only be substituted by other prepattern variables, while 
ordinary variables can be substituted by arbitrary terms (which include other prepattern variables).
In a well-typed signature, if $r : A = M$ is a recursion declaration, then
$A$ consists of purely prepattern $\Pi$-abstractions (judgment $A \isprepattern$)
and for all $r \cdot S$ in the signature,  $S$ consists of purely prepattern applications and is thus called a prepattern spine (judgment $S \isprepattern$).
The prepattern variables are similar to those introduced by the $\nabla$-operator \cite{Miller05tocl}, which models the concept of fresh names, but 
here in a dependently typed setting, types may depend on prepattern variables.

In an actual implementation, 
the usages of prepattern types may impose additional burdens on 
the programmer. As a remedy,
the implementation could infer which variables are prepattern variables 
based on whether they appear as arguments to 
recursion constants and propagate such information. 

\begin{figure}[t]
    \begin{longtable}{lrcl}
        Signatures & $\Sigma $ & $::= $ & $ \cdot \mid \Sigma, a : K \mid \Sigma, c : A \mid \Sigma, r : A = M  $ \\
        Contexts & $\Gamma$ & $::= $ & $\cdot \mid \Gamma, x : A \mid \Gamma, x \prepatof A$  \\
        Kinds & $K$ & $ ::= $ & $\type \mid \cotype \mid \Pi x : A.\ K \mid \Pi x \prepatof A. \, K$ \\
        Canonical types & $A, B$ & $ ::= $ & $P \mid \Pi x : A_2.\ A_1 \mid \Pi x \prepatof A_2. \, A_1$ \\
        Atomic types & $P$ & $ ::= $ & $a \cdot S$ \\
        Canonical terms & $M$ & $ ::= $ & $ R  \mid \lambda x.\, M $\\
        Neutral terms & $R$ & $ ::= $ & $ H \cdot S  $\\
        Heads & $H$ & $ ::= $ & $ x  \mid c  \mid r  $\\
        Spines & $S$ & $ ::= $ & $ M ; S\mid \prepatternvar x ; S \mid \snil $\\
    \end{longtable}
    \caption{The Syntax for $\CoLFCircular$}
    \label{fig:syntax_for_colf_two_circular}
\end{figure}


\subsection{Trace Condition}
\label{subsection:trace_condition}
In a signature $\Sigma$, we say that a type $A$ is inductive if $A = \Pi x_1\dots \Pi x_n:A_n. a\cdot S$ and $a : \Pi y_1 \dots \Pi y_m : B_m.\, \type$, 
and a type $A$ coinductive if $A = \Pi x_1\dots \Pi x_n:A_n. a\cdot S$ and $a : \Pi y_1 \dots \Pi y_m : B_m.\, \cotype$.
A constructor $c$ is inductive if $c : A \in \Sigma$ and $A$ is inductive, and $c$ is coinductive if $c : A \in \Sigma$ and $A$ is coinductive.


The validity of the terms is enforced through a trace condition \cite{Fortier13csl,Brotherston11jlc} on cycles.
A trace is a sequence of constructor constants or variables, where each constructor or variable is 
a child of the previous one.
A trace from a recursion constant $r$ to itself is a sequence starting with
the head of the definition of $r$ and ending with the parent of an occurrence of $r$. 
In Fig.~\ref{fig:sigantures_for_section_motivation}, 
a trace from \verb$p2$ to itself is $[\mt{pad}]$, and 
a trace from \verb$s1$ to itself is $[\mt{cocons}, \mt{pad}, \mt{pad}, \mt{next}]$.
Traces cross into definitions of recursion constants.
Thus, a trace from \verb$p6$ to itself is $[\mt{pad}, \mt{pad}]$, which is also 
a trace from \verb$p7$ to itself.
A trace from \verb$s4$ to itself is $[\mt{cocons}, \mt{next}]$, and 
a trace from \verb$p5$ to itself is $[\mt{next}, \mt{cocons}]$.
If $r = \lambda x. f\, (r\, x)\, (g\, (r\, x))$ (more precisely
$r = \lambda x.\, f\cdot (r\cdot([x]); g\cdot(r \cdot ([x])))$), then there are two traces from $r$ to itself, i.e.,
$[f]$ and $[f,g]$.

A higher-order rational term $M$ is \emph{trace-valid} if for all 
recursion constants $r$ in $M$, each trace from $r$ to itself 
contains a coinductive constructor, and that coinductive constructor has the highest priority 
among all constructors on that trace.
To ensure trace validity, it is sufficient to check in a recursive definition, 
all occurrences of recursion constants are \emph{guarded by}
some coinductive constructor of the highest priority.
The guardedness condition (judgment $\vdash_\Sigma r \guardedin M$) means that occurrences of $r$ in $M$ are guarded by some coinductive constructor
of the highest priority, and the condition is decidable. In a well-typed signature $\Sigma$, if $r : A= M \in \Sigma$, then $\vdash_\Sigma r\guardedin M$.
A detailed algorithm for checking trace-validity is presented in 
Appendix~\ref{appendix:guardedness_checking}\oftheextendedversion.
The reader may check guardedness for all valid terms in Fig.~\ref{fig:sigantures_for_section_motivation}.



\subsection{Hereditary Substitution}

Hereditary substitution \cite{Watkins02tr,Harper07jfp} provides a method of substituting one 
canonical term into another and still get a canonical term as the output by
performing type-based normalization. This technique simplifies the definition of the term 
equality in the original LF \cite{Harper93jacm,Harper05tocl} by separating the term equality and normalization
from type checking. 
We extend the definition of hereditary substitution to account for 
recursion constants.
Hereditary substitution is a partial operation on terms. 
When input term is not well-typed or prepattern restriction is not respected, 
the output may be undefined.

Hereditary substitution takes as an extra argument the simple type of the term being substituted by.
The simple type $\tau$ is inductively 
generated by the following grammar. 
\[\tau ::= * \mid \tau_1 \to \tau_2\]

We write $\erased{A}$ for the simple type that results from erasing dependencies in $A$. 
We write $[N/x]^\tau M$ for hereditarily substituting $N$ for free  ordinary variable $x$ in $M$.
The definition proceeds by induction on  $\tau$ and the structure of $M$. 
For prepattern variables, since they may only stand for other prepattern variables, we use 
a notion of renaming substitution.  The renaming substitution $\varrenamesubst{y/x}M$ renames
a prepattern variable or an ordinary variable $x$ to prepattern variable $y$ in $M$.
Both substitutions naturally extend to other syntactic kinds. 
Hereditary substitution relies on renaming substitution when reducing prepattern applications.
Because of the prepattern restriction, recursion constants are only applied to prepattern variables in a well-formed signature, 
and we never substitute into a recursive definition.
Let $\sigma$ be a simultaneous renaming substitution, a notion generalized from renaming substitutions, 
we write $\varrenamesubst{\sigma}M$ for carrying out substitution $\sigma$ on $M$.

The definition for hereditary substitution is shown in Fig.~\ref{fig:hereditary_substitution_essential_cases}. 
Appendix~\ref{appendix:hereditary_substitution_and_renaming_substitution}\oftheextendedversion[ ]contains other straightforward cases of the definition. 
We note that prepattern $\Pi$-types erase to a base type $*$ 
because we may only apply terms of prepattern $\Pi$-types to prepattern variables, and thus the structure of the argument term
does not matter.

\begin{figure}[t]
\begin{flushleft}
\begin{multicols}{2}
    \boxed{A^o = \tau}  \\
    $(\pitype{x}{A_2}{A_1})^o = (A_2^o) \to (A_1^o)$ \\
    $(\prepatpi{x}{A_2}{A_1})^o = * \to (A_1^o)$ \\
    $(P)^o = *$ \\
    \boxed{[N/x]^\tau M = M'} \\
    $[N/x]^\tau R = [N/x]^\tau R$ \\
    $[N/x]^\tau (\lambda y.M) = \lambda y. [N/x]^\tau M$,    $y \ne x $ \\

    \boxed{[N/x]^\tau R = R'} \\
    $[N/x]^\tau( x \cdot S) = ([N/x]^\tau S) \rhd^{\tau} N$ \\
    $[N/x]^\tau( y \cdot S) =  y \cdot ([N/x]^\tau S) $,   $ y \ne x$\\
    $[N/x]^\tau( c \cdot S) = c \cdot ([N/x]^\tau S) $\\
    $[N/x]^\tau( r \cdot S) = r \cdot ([N/x]^\tau S) $\\
    \boxed{[N/x]^\tau  S = S'} \\
    $[N/x]^\tau\snil = \snil$ \\
    $[N/x]^\tau(M;S) = ([N/x]^\tau M); ([N/x]^\tau S)$ \\
    $[N/x]^\tau(\prepatternvar x;S) = \m{undefined}$   
    \\
    $[N/x]^\tau(\prepatternvar z;S) = \prepatternvar z; ([N/x]^\tau S)$,   $x \ne z$ \\
    \boxed{S \rhd ^\tau N = R'} \\
    $\snil \rhd ^* R = R$ \\
    $(N; S) \rhd^{\tau_2 \to \tau_1} (\lambda x. M) = S \rhd^{\tau_1} ([N/x]^{\tau_2}M)$ \\
    $(\prepatternvar y; S) \rhd^{* \to \tau_1} (\lambda x. M) = S \rhd^{\tau_1} (\varrenamesubst{y/x}M)$ \\
\end{multicols}
\end{flushleft}
\caption{Hereditary Substitutions}
\label{fig:hereditary_substitution_essential_cases}
\end{figure}


 \subsection{Term Equality}

 \label{subsection:term_equality}

 The equality checking of circular terms is carried out by
 iteratively unfolding recursive definitions
 \cite{Amadio93toplas,Brandt98fi,Danielsson10mpc,Ligatti17toplas}.
 The algorithm here is a slight adaptation of
 the equality algorithm for regular B\"{o}hm trees by Huet \cite{Huet98mscs}, tailored 
 to the specific case of $\CoLFCircular$'s canonical term syntax.
 We emphasize that the equality algorithm
 can treat terms that are not trace-valid or well-typed, 
 and is thus decoupled from validity checking and type checking.
 The algorithm itself checks for the prepattern restriction on recursion constants and contractiveness condition on recursive definitions.
 These checks are essential to ensure termination in the presence of forward referencing inside recursive definitions.

 We define the judgment $\Delta ; \Theta\vdash_\Sigma M = M'$ to mean $M$ and $M'$, with free variables from $\Theta$, are equal under the assumptions $\Delta$, 
 with consideration of recursive definitions in $\Sigma$. The variable list $\Theta$ is similar to $\Gamma$ except it doesn't have the types for the variables. It is merely a list of pairwise distinct variables. 
 Similarly, we define the judgment $\Delta ; \Theta\vdash_\Sigma S = S'$ to mean spines $S$ and $S'$ are element-wise equal.
 Equalities in $\Delta$ will be of the form $\ctxeqassump{\Theta}{M}{M'}$ where $\Theta$ holds free variables of $M$ and $M'$.
 We write $\Theta \vdash M$ to mean that $FV(M) \subseteq \Theta$.
 We define simultaneous variable renaming, that
$\sigma$ is a variable renaming from $\Theta'$ to $\Theta$, written
 $\Theta \vdash \sigma: \Theta'$
 to mean that 
 if $\Theta' \vdash M$, then $\Theta \vdash \varrenamesubst{\sigma}{M}$. 
 For instance, if we have $x \vdash \varrenamesubst{x/y, x/z} : y, z$ and $y, z \vdash y \cdot \prepatternvar z$, then $x \vdash \varrenamesubst{x/y, x/z}(y \cdot \prepatternvar z)$, i.e., $x \vdash x \cdot \prepatternvar x$. 
 The rules for the judgments are presented in Fig.~\ref{fig:higher_order_rational_terms_equality_checking}.
 Recall that $M$ is contractive ($M \iscontractive$) if the head of $M$ is not a recursion constant.

 \begin{figure}[t]
 \boxed{\Delta; \Theta\vdash_\Sigma M = M'}
 \begin{mathparpagebreakable}
 \inferrule{
     \Theta \vdash \sigma : \Theta'
 }{
     \Delta, \ctxeqassump{\Theta'}{H\cdot S_1}{ H'\cdot S_2}; \Theta\vdash_\Sigma 
      \varrenamesubst{\sigma}{(H\cdot S_1)} = \varrenamesubst{\sigma}{(H'\cdot S_2)}
 }(1)
 
 \inferrule{
     r : A = M \in \Sigma
     \\
     S_1 \isprepattern
     \\
     M \iscontractive
     \\
     \Delta, 
     \ctxeqassump{\Theta}{ r\cdot S_1 }{ H\cdot S_2 }
     ; \Theta
     \vdash_\Sigma
     S_1 \rhd ^{A^o} M
     =H \cdot S_2
 }{
     \Delta; \Theta\vdash_\Sigma r\cdot S_1 = H\cdot S_2
 }(2)
 
 \inferrule{
     r : A = M \in \Sigma
     \\
     S_2 \isprepattern
     \\
     M \iscontractive
     \\
     H \ne r'
     \\
     \Delta, \ctxeqassump{\Theta}{ H\cdot S_1 }{ r\cdot S_2 }
     ; \Theta
     \vdash_\Sigma 
     H\cdot S_1 = 
     S_2 \rhd ^{A^o} M
 }{
     \Delta; \Theta\vdash_\Sigma H\cdot S_1 = r\cdot S_2
 }(3)
 
 \inferrule{
      \Delta; \Theta \vdash_\Sigma  S = S'
 }{
     \Delta; \Theta \vdash_\Sigma    c \cdot S = 
      c \cdot S'
 }(4)
 \hspace{1em}
 \inferrule{
      \Delta; \Theta \vdash_\Sigma  S = S'
 }{
     \Delta; \Theta \vdash_\Sigma    y \cdot S = 
      y \cdot S'
 }(5)
 \hspace{1em}
 \inferrule{
     \Delta; \Theta, x \vdash_\Sigma  M = M'
 }{
     \Delta; \Theta \vdash_\Sigma  \lambda x.\, M = \lambda x.\, M'
 }(6)
 \end{mathparpagebreakable}

 \boxed{\Delta; \Theta\vdash_\Sigma S = S'}
 \begin{mathparpagebreakable}
 \inferrule{
 }{
     \Delta; \Theta \vdash_\Sigma \snil = \snil
 }
 \hfill
 \inferrule{
     \Delta; \Theta \vdash_\Sigma M = M'
     \\
     \Delta; \Theta \vdash_\Sigma S = S'
 }{
     \Delta; \Theta \vdash_\Sigma M ; S = M'; S'
 }
 \hfill
 \inferrule{
     \Delta; \Theta \vdash_\Sigma S = S'
 }{
     \Delta; \Theta \vdash_\Sigma \prepatternvar x ; S = \prepatternvar x; S'
 }
 \end{mathparpagebreakable}

 \caption{Equality Checking}
 \label{fig:higher_order_rational_terms_equality_checking}
\end{figure}

\subsubsection{An Example.}

Assume the signature in Section \ref{subsection:finitely_padded_streams}, and
consider a stream generator that repeats its arguments. 
The stream may be represented by terms \verb$r1$ and \verb$r2$ below.
Note that in the concrete syntax, square brackets represent $\lambda$-abstractions.
\begin{verbatim}
r1 : nat -> pstream = [x] cocons x (next (r1 x)).
r2 : nat -> pstream = [x] cocons x (next (cocons x (next (r2 x)))).
\end{verbatim}
Because \verb#r1# is a recursion constant,
its type is a prepattern-$\Pi$ type, and this restriction is respected in the body as $x$ is 
a prepattern variable. 

We want to show that \verb$r1$ and \verb$r2$ are equal in the framework.
Let $\Sigma$ be the signature of Section \ref{subsection:finitely_padded_streams} plus the definitions for 
\verb$r1$ and \verb$r2$.
We illustrate the process of checking that $;\vdash_{\Sigma} \lambda x.\, \mt{r1} \cdot (\prepatternvar x) = \lambda x.\, \mt{r2} \cdot(\prepatternvar x)$ as a 
search procedure for a derivation of this judgment, where initially
both $\Delta$ and $\Theta$ are empty. 

Immediately after rule (6) we encounter $; x \vdash_{\Sigma} \mt{r1} \cdot(\prepatternvar x) = \mt{r2}\cdot ( \prepatternvar x)$, we memoize 
this equality by storing  ${\ctxeqassump{x}{\mt{r1} \cdot(\prepatternvar x)}{\mt{r2}\cdot ( \prepatternvar x)}} $ in $\Delta$ as in rule (2), and 
unfold the left-hand side. 
Then we proceed with the judgment.
\[
    {\ctxeqassump{x}{\mt{r1} \cdot(\prepatternvar x)}{\mt{r2}\cdot ( \prepatternvar x)}} ; x 
\vdash_\Sigma \mt{cocons} \cdot (x; \mt{next} \cdot (\mt{r1} \cdot( \prepatternvar x))) 
=  \mt{r2} \cdot (\prepatternvar x) \]
We then use rule (3) to unfold the right-hand side and store 
then current equation
 in the context.
Then after several structural rules, we have
\[{\ctxeqassump{x}{\mt{r1} \cdot(\prepatternvar x)}{\mt{r2}\cdot ( \prepatternvar x)}}, \dots; x \vdash_\Sigma \mt{r1} \cdot( \prepatternvar x) =  \mt{cocons}\cdot (x; \mt{next}\cdot(\mt{r2} \cdot (\prepatternvar x)))\]
At this point, rule (2) applies. 
We add the current equation
to the context and unfold the left recursive definition. Then after several structural rules, 
we encounter the following judgment.
\[
{\ctxeqassump{x}{\mt{r1} \cdot(\prepatternvar x)}{\mt{r2}\cdot ( \prepatternvar x)}}
, \dots; x \vdash_\Sigma \mt{r1} \cdot( \prepatternvar x) =  \mt{r2} \cdot (\prepatternvar x)\]
Now we can close the derivation with rule (1) using identity substitution.

\subsubsection{Decidability.}
 Huet \cite{Huet98mscs} has proved the termination, soundness, and completeness in the case of untyped regular B\"ohm trees.
 Our proof shares the essential idea with his proof.
 The termination relies on the fact that terms only admit finitely many subterms modulo renaming of both free 
 and bound variables, and only subterms will appear in $\Delta$. The soundness and completeness are proved with respect 
 to the infinite B\"ohm tree \cite{Barendregt85book} generated by unfolding the terms indefinitely, which again corresponds to 
 a bisimulation between terms.

 \begin{theorem}
     [Decidability of Term Equality]
    \label{thm:decidability_of_term_equality}
     It is decidable whether $\Delta; \Theta\vdash_\Sigma M = M'$ for any rational term $M$ and $M'$.
 \end{theorem}

 \begin{proof}
    We first show that there is a limit on the number of equations in $\Delta$. 
    Then the termination follows the lexicographic order of the assumption capacity (difference 
    between current number of assumptions in $\Delta$ and the maximum), and the structure of 
    the terms under comparison.
    It is obvious that rules (4)(5)(6) decompose the structure of the terms and rules (2)(3) 
    reduce assumption capacity. It remains to show that the size of $\Delta$ has a limit.

     The prepattern conditions on rules (2)(3) ensure that the expansion of recursive definitions will only involve renaming substitutions, 
     and thus the resulting term will be an $\alpha$-renaming of the underlying definition. No structurally new terms will be produced
     as a result of renaming substitution in rules (2)(3). We construct a finite set of all possible terms 
     that could be added to the context. Each term is of finite depth and breadth limited by the existing constructs 
     in the signature, and consists of finitely many constants, variables, 
     and recursion constants. The constants and recursion constants are limited to those already presented in the signature. 
     Although there are infinitely many variables, there are finitely many terms of bounded depth and width that are distinct
     modulo renaming of both bound and free variables. Thus, the set of terms that can appear as an element of $\Delta$ 
     is finite, modulo renaming of free variables. The estimate of a rough upper bound can be found in 
     Appendix~\ref{appendix:equality_upper_bound_estimate}\oftheextendedversion.

 \end{proof}
 
We specify the infinite unfolding by 
specifying its unfolding to a B\"ohm tree of depth $k$, which is a finite approximation 
to the infinite B\"ohm tree,
for each $k \in \bN$.
Then the infinite B\"ohm tree is limit of all its finite approximations.
We use the judgment $\exp\pdepth{k}(M)=\pdepth{k}M'$ to denote the expansion
of a higher-order rational term $M$ to a B\"ohm tree $M'$  of depth $k$, 
and use the judgment $\exp(N) = N'$ to express that the higher-order rational term $M$ 
expands to infinite B\"ohm tree $N'$.
We also enrich the syntax of B\"ohm trees with prepattern variables.
    The full set of expansion rules can be found in 
    Appendix~\ref{appendix:expansion_as_bohm_trees}\oftheextendedversion.
All cases are structural except 
for the following case when we expand a recursion constant, where we look up the definition 
of the recursion constant and plug in the arguments.
    \[ \exp\pdepth {k+1}( r \cdot S)  = \pdepth{k+1}   \exp\pdepth{k+1}(S \rhd^{A^o} M) \text{ if $r : A = M \in \Sigma$ and $S \isprepattern$}
    \]

\begin{lemma}
    [Expansion Commutes with Hereditary Substitution]
    For all $k$, $\tau$, $M$ and $N$, 
    $\exp\pdepth k([N/x]^\tau M) =\pdepth{k} [\exp\pdepth{k}(N)/x]^\tau(\exp\pdepth k (M))$ 
    if defined.
\end{lemma}

\begin{proof}
    Directly by lexicographic induction on $k$ and the structure of $M$. 
\end{proof}

\begin{theorem} [Soundness of Term Equality]

    If $\cdot ; \Theta \vdash M = M'$,
    then $\exp\pdepth k(M) =\pdepth k \exp\pdepth k(M')$ for all $k$.
\end{theorem}

\begin{proof}
    By lexicographic induction on the depth $k$ and the derivation $\Delta ; \Theta\vdash M = M'$.
    The case for the rule (1)  is immediate by applying renaming substitutions at the closure rule.
     The cases for rules (2)(3) follow from the commutation lemma. 
    The cases for rules (4)(5)(6) follow from the definition of $\exp$.
\end{proof}

\begin{theorem}
    [Completeness of Term Equality]

    For rational terms $M$ and $M'$, 
    with free variables from $\Theta$,
     if $\exp(M) = \exp(M')$, then  $\cdot; \Theta\vdash M = M'$.
\end{theorem}

\begin{proof}
    The equality algorithm is syntax-directed. 
    We construct the derivation of  $\cdot; \Theta\vdash M = M'$ by syntax-directed proof search
    following the structure of $M$. 
    Every trace of $\exp(M)$ and $\exp(M')$ 
   corresponds to a trace in the derivation of $\cdot; \Theta\vdash M = M'$.
    If $\exp(M) = \exp(M')$, then two terms are equal on every trace, and 
    there will be exactly one rule that applies at every point in the construction of the equality derivation.
    Termination is assured by Theorem~\ref{thm:decidability_of_term_equality}. 
\end{proof}


\subsection{Type Checking Rules}

For type checking, we define the judgments in Fig.~\ref{fig:type_checking_judgments} 
by simultaneous induction.  Because 
recursion constants may be forward referenced, 
we need to have access to later declarations that have not been checked during 
the checking of earlier declarations. In order to ensure the otherwise linear order of the 
declarations, the type checking judgments are parametrized by a pair of signatures $\rsignature$, where $\Xi$ is the local signature 
that contains type-checked declarations before the current declaration 
 and $\Sigma$ is the global signature that contains 
full signatures, including declarations that have not been checked.
In particular, recursion constants available for forward-referencing will be in $\Sigma$ but not $\Xi$.
The type equality judgments $\Gamma \vdash_\Sigma A_1 = A_2$, $\Gamma \vdash_\Sigma P_1 = P_2$ 
only need to read recursive definitions from the global signature, and 
do not need to access the local signature.

\begin{figure}[t]
    \begin{longtable}{ll}
        \label{table:syntax_of_colf_two_circular}
         $\Sigma\validsig$ & Signature $\Sigma$ is type correct categorically \\
         $\vdash_{\Sigma}\Xi\validsig$ & Local signature $\Xi$ is type correct with global signature $\Sigma$\\
         $\vdash_\rsignature\Gamma\validctx$ & Context $\Gamma$ is well-formed \\
          $\Gamma\vdash_\rsignature K \Leftarrow \kind$ & Kind $K$ is a valid kind \\
          $\Gamma\vdash_\rsignature A \Leftarrow \typecotype$ & Type $A$ is a canonical type \\
          $\Gamma\vdash_\rsignature P \Rightarrow K$ & Atomic type $P$ synthesizes kind $K$ \\
          $\Gamma \vdash_\rsignature S \rhd K \Rightarrow K'$ & Spine $S$ applied to kind $K$ produces kind $K'$ \\
        $\Gamma \vdash_\rsignature M \Leftarrow A$ & Term $M$ checks against type $A$ \\
          $\Gamma \vdash_\rsignature R \Rightarrow P$ & Neutral term $R$ synthesizes type $P$  \\
          $\Gamma \vdash_\rsignature S \rhd A \Rightarrow P$ & Spine $S$ applied to canonical type $A$ produces atomic type $P$ \\
        $\Gamma \vdash_\Sigma A_1 = A_2$ & Types $A_1$ and $A_2$ are equal canonical types \\
        $\Gamma \vdash_\Sigma P_1 = P_2$ & Types $P_1$ and $P_2$ are equal atomic types \\
    \end{longtable}
    \caption{Type Checking Judgments}
    \label{fig:type_checking_judgments}
\end{figure}

    A selection of type checking rules that are essential are presented in Fig.~\ref{fig:type_checking_rules_selection}.
    The rest of the rules  can be found in 
    Appendix~\ref{appendix:type_checking_rules_for_colf}\oftheextendedversion.
    To ensure the correct type checking order, i.e., the body of a recursive definition is checked after
    the types of all recursion constants within are checked, we defer checking the body of all recursive definitions 
    to the end. This approach is viable because the term equality algorithm soundly terminates even when the recursive definition
    is not well-typed. For instance, if the signature $\Sigma = c_1 : A_1, c_2 : A_2, r_1 : A_3 = M_1, c_3 : A_4, r_2 : A_5 = M_2 $, 
    then the order of checking is $A_1, A_2, A_3, A_4, A_5, M_1, M_2$. This order is expressed 
    in the type checking rules by an annotation on specific premise of the rules.
    The annotation $   [{
        \vdash_{\rsignature} M \Leftarrow A
        }]^{\m{1:deferred}}$ means that this judgment is to be checked after all the typing judgments have been checked.
        That is, when we check this premise, we have checked that $\vdash _\Sigma \Sigma \validsig$. 
        Because of the deferred checking of recursive definitions, 
        the judgment $\vdash_\Sigma \Xi \validsig$ does not require 
         the body of recursion declarations in $\Xi$ to be well-typed. 
         However, the categorical judgment $\Sigma \validsig$ requires the body 
         of every recursion declaration to be well-typed.

        To enforce the restriction that forward references only happen in a recursive definition, 
        the annotation $    [\text{or $r : A = M \in \Sigma$}] ^ {\m{2:definitions}}$ means 
        that forward reference only occurs during the checking of recursive definitions (which are deferred) and nowhere else. 

\begin{figure}
\begin{multicols}{2}
$\boxed{\Sigma \validsig}$ 
\begin{mathparpagebreakable}
        \inferrule{
            \vdash_\Sigma \Sigma \validsig 
        }{
            \Sigma \validsig
        }
\end{mathparpagebreakable}

$\boxed{\vdash_{\Sigma}\Xi \validsig}$ 
\begin{mathparpagebreakable}
        \inferrule{
        }{
         \vdash_{\Sigma}    \cdot \validsig
        }
    
        \inferrule{
           \vdash_{\Sigma} \Xi \validsig 
            \\ 
            \vdash_{\rsignature} K \Leftarrow \kind
        }{
           \vdash_{\Sigma} \Xi, a : K \validsig
        }
    
        \inferrule{
            \vdash_{\Sigma}\Xi \validsig 
            \\ 
            \vdash_{\Sigma} A \Leftarrow \typecotype
        }{
            \vdash_{\Sigma}\Xi, c : A \validsig
        }

        \inferrule{
            \vdash_{\Sigma}\Xi \validsig 
            \\ 
            \vdash_{\rsignature} A \Leftarrow \typecotype
            \\
            [{
            \vdash_{\rsignature} M \Leftarrow A
            }]^{\m{1:deferred}}
            \\
            A \isprepattern
            \\
            M \iscontractive
            \\
            \vdash _\Sigma r \guardedin M 
        }{
            \vdash_{\Sigma}\Xi, r : A = M \validsig
        }
    \end{mathparpagebreakable}
    
    $\boxed{\Gamma \vdash_\rsignature K \Leftarrow \kind}$
    \begin{mathpar}
        \inferrule{
        }{
            \Gamma \vdash_\rsignature \type \Leftarrow \kind
        }
    
        \inferrule{
        }{
            \Gamma \vdash_\rsignature \cotype \Leftarrow \kind
        }

        \inferrule{
            \Gamma \vdash_\rsignature A \Leftarrow \typecotype
            \\
            \Gamma, x \prepatofoption A\vdash_\rsignature K \Leftarrow \kind
        }{
            \Gamma \vdash_\rsignature  \Pi x \prepatofoption A . \, K \Leftarrow \kind
        }
    \end{mathpar}
    
    $\boxed{\Gamma \vdash_\rsignature A \Leftarrow \typecotype}$ 
    \begin{mathpar}
        \inferrule{
            \Gamma \vdash_\rsignature A_2 \Leftarrow \typecotype\\
            \Gamma, x \prepatofoption A_2\vdash_\rsignature A_1 \Leftarrow \typecotype
        }{
            \Gamma \vdash_\rsignature \Pi x \prepatofoption A_2. \, A_1 \Leftarrow \typecotype
        }
    
        \inferrule{
            \Gamma \vdash_\rsignature P \Rightarrow K \\
            K = \type/\cotype
        }{
            \Gamma \vdash_\rsignature P \Leftarrow \typecotype
        } 
    \end{mathpar}

    \boxed{\Gamma\vdash_\rsignature P \Rightarrow K}
    \begin{mathpar}
        \inferrule{
            a : K \in \Xi
            \\
            \Gamma \vdash _\rsignature S \rhd K \Rightarrow K'
        }{
            \Gamma\vdash_\rsignature a\cdot S \Rightarrow K'
        }
    \end{mathpar}

    \boxed{\Gamma \vdash_\rsignature S \rhd K \Rightarrow K'} 
    \begin{mathpar}
    \inferrule{
    }{
        \Gamma  \vdash_\rsignature \snil  \rhd K \Rightarrow K
    }

    \inferrule{
        \Gamma \vdash_\rsignature M \Leftarrow A_2 \\
        [M/x]^{\erased{A_2}}K = K'
        \\
        \Gamma \vdash_\rsignature S \rhd K' \Rightarrow K''
    }{
        \Gamma  \vdash_\rsignature M;S  \rhd \Pi x : A_2. \, K \Rightarrow  K''
    }

    \inferrule{
        y \prepatof A_2' \in \Gamma\\
        \Gamma \vdash_\rsignature A_2' = A_2
        \\
        \varrenamesubst{y/x}{K} = K'
        \\
        \Gamma \vdash_\rsignature S \rhd K' \Rightarrow K''
    }{
        \Gamma  \vdash_\rsignature\prepatternvar  y;S  \rhd \Pi x \prepatof A_2. \, K \Rightarrow  K''
    }
    \end{mathpar}

    \boxed{\Gamma \vdash_\rsignature M \Leftarrow A} 
    \begin{mathpar}
    \inferrule{
        \Gamma  \vdash_\rsignature R \Rightarrow P'
        \\
        \Gamma \vdash_\Sigma P' = P
    }{
        \Gamma  \vdash_\rsignature R \Leftarrow P
    }

    \inferrule{
        \Gamma, x \prepatofoption A_2 \vdash_\rsignature M \Leftarrow A_1
    }{
        \Gamma  \vdash_\rsignature \lambda x.\, M \Leftarrow \Pi x \prepatofoption A_2.\,A_1
    }
    \end{mathpar}

    \boxed{\Gamma \vdash_\rsignature R \Rightarrow P} 
    \begin{mathpar}
    \inferrule{
        (c/x : A \in \Gamma
        \text{ or }
        x \prepatof A \in \Gamma)
        \\
        \Gamma \vdash _\rsignature S \rhd A \Rightarrow P
    }{
        \Gamma  \vdash_\rsignature c/x  \cdot S \Rightarrow P
    }

    \inferrule{
        r : A = M \in \Xi
        \\
        [\text{or $r : A = M \in \Sigma$}] ^ {\m{2:definitions}}
        \\
        \Gamma \vdash _\rsignature  S \rhd A \Rightarrow P
    }{
        \Gamma  \vdash_\rsignature r \cdot  S \Rightarrow P
    }
    \end{mathpar}

    \boxed{\Gamma \vdash_\rsignature S \rhd A \Rightarrow P} 
    \begin{mathpar}
    \inferrule{
    }{
        \Gamma  \vdash_\rsignature \snil  \rhd P \Rightarrow P
    }

    \inferrule{
        \Gamma \vdash_\rsignature M \Leftarrow A_2 \\
        [M/x]^{\erased{A_2}}A_1 = A_1'
        \\
        \Gamma \vdash_\rsignature S \rhd A_1' \Rightarrow P
    }{
        \Gamma  \vdash_\rsignature M;S  \rhd \Pi x : A_2. \, A_1 \Rightarrow  P
    }

    \inferrule{
        y \prepatof A_2' \in \Gamma\\
        \Gamma \vdash_\rsignature A_2' = A_2
        \\
        \varrenamesubst{y/x}{A_1} = A_1'
        \\
        \Gamma \vdash_\rsignature S \rhd A_1' \Rightarrow P
    }{
        \Gamma  \vdash_\rsignature \prepatternvar y;S  \rhd \Pi x \prepatof A_2. \, A_1 \Rightarrow  P
    }
    \end{mathpar}
\end{multicols}
\caption{Type Checking Rules (Condensed Selection)}
\label{fig:type_checking_rules_selection}
\end{figure}

\subsection{Metatheorems}
\label{subsection:metatheorems_after_type_checking_rules}
We state some properties about hereditary substitution and type checking.

\begin{theorem}
    [Hereditary Substitution Respects Typing]

    Given a checked signature $\Sigma$ where $\Sigma \validsig$, 
     if $\Gamma \vdash_\rsignature N \Leftarrow A$ and $\Gamma, x : A, \Gamma' \vdash M \Leftarrow B$, then 
    $\Gamma, [N/x]\uperased{A}\Gamma'\vdash_\rsignature [N/x]^{A^o} M \Leftarrow [N/x]^{A^o}B$.
\end{theorem}

\begin{proof}
    By induction on the second derivation, with similar theorems for other judgment forms.
    This proof is similar to those in \cite{Watkins02tr,Harper07jfp}. Because 
    of the prepattern restriction, hereditary substitutions do not 
    occur inside recursive definitions and is thus similar to hereditary substitutions in LF.
\end{proof}

\begin{theorem}
    [Decidability of Type Checking]

    All typing judgments are algorithmically decidable.
\end{theorem}

\begin{proof}
    The  type checking judgment is syntax directed. Hereditary substitutions are defined by induction on the erased 
    simple types and always terminate. Equality of types ultimately reduces to equality of terms, and we
     have proved its termination in Section~\ref{subsection:term_equality}.
\end{proof}

\section{Encoding Subtyping Systems for Recursive Types}

In the presentation of case studies,
we use the concrete syntax of our implementation, following Twelf \cite{Pfenning98guide}.
 The prepattern annotations 
are omitted. The full convention can be found in 
Appendix~\ref{appendix:concrete_syntax_for_colf}\oftheextendedversion.
Representations of circular derivations involve dependent usages of $\cotype$'s.

\label{section:encoding_subtyping_systems_for_recursive_types}
\subsection{Encoding a Classical Subtyping System}

We present a mixed inductive and coinductive definition of subtyping using Danielsson and Altenkirch's \cite{Danielsson10mpc} subtyping 
system.  The systems concern the subtyping of types given by the following grammar.
\newcommand{\arrto}{\rightarrowtriangle}
\[\tau ::= \bot \mid \top \mid \tau_1\arrto\tau_2 \mid \mu X. \tau_1 \arrto \tau_2 \mid X\]
The subtyping judgment is defined by five axioms and two rules, 
The axioms are
\begin{enumerate}
    \item $\bot\le\tau$ ($\m{bot}$)
    \item $\tau \le\top$($\m{top}$)
    \item $ \mu X.\tau_1\to\tau_2
        \le 
        [\mu X.\tau_1 \to\tau_2/X](\tau_1\to\tau_2) $
    ($\m{unfold}$)
    \item
    $ [\mu X.\tau_1\to\tau_2/X](\tau_1\to\tau_2) 
        \le 
        \mu X.\tau_1\to\tau_2
        $
    ($\m{fold}$)
    \item $\tau \le \tau$ ($\m{refl}$)
\end{enumerate}
And the rules are shown below,
where $\m{arr}$ is coinductive and is written using a double horizontal line, and $\m{trans}$ is inductive.
The validity condition of mixed induction and coinduction entails that a derivation consisting purely of $\m{trans}$ 
rules is not valid. 
\begin{mathparpagebreakable}
    {
    \mprset{fraction={===}}
    \inferrule{
        \tau_1 \le \sigma_1
        \\
        \sigma_2 \le \tau_2
    }{ \sigma_1 \to \sigma_2 \le \tau_1 \to \tau_2} (\m{arr})
    }

    \inferrule{
        \tau_1 \le \tau_2
        \\
        \tau_2 \le \tau_3
    }{
        \tau_1 \le \tau_3
    }(\m{trans})
\end{mathparpagebreakable}



Danielsson and Altenkirch defined the rules using Agda's mixed inductive and coinductive datatype (shown in 
 Appendix~\ref{appendix:subtyping_proofs_agda_colf}\oftheextendedversion) and the encoding in $\CoLFCircular$ 
is shown in Fig.~\ref{fig:encoding_of_subtyping_in_colf}.
The curly brackets indicate explicit $\Pi$-abstractions and the free capitalized variables are implicit $\Pi$-abstracted.
 We note that the mixed inductive and coinductive nature of the 
subtyping rules
reflected in $\CoLFCircular$ as two predicates, the inductive \verb#subtp# and the coinductive \verb#subtpinf#, and the latter has a higher priority.
Clauses defining one predicate refer to the other predicate as a premise, e.g. \verb#subtp/arr# and \verb#inf/arr#. 
Let $\enc{-}$ denote the encoding relation, and 
 we have $\enc{\mu X. \sigma \arrto \tau} = \mt{mu}\, \enc{X.\sigma} \, \enc{X.\tau}$.

 \begin{figure}[t]
\begin{minipage}{0.5\linewidth}
\begin{verbatim}
tp : type.
bot : tp.
top : tp. 
arr : tp -> tp -> tp.
mu : (tp -> tp) -> (tp -> tp) -> tp.
\end{verbatim}
\end{minipage}
\begin{minipage}{0.5\linewidth}
\begin{verbatim}
subtp : tp -> tp -> type.
subtpinf : tp -> tp -> cotype.
subtp/top : subtp T top.
subtp/bot : subtp bot T.
refl : subtp T T.
\end{verbatim}
\end{minipage}

\vspace{.1em}
\begin{minipage}{1.0\linewidth}
\begin{verbatim}
trans : subtp T1 T2 -> subtp T2 T3 -> subtp T1 T3.
subtp/arr : subtpinf T1 T2 -> subtp T1 T2.
unfold : {T1}{T2} subtp (mu T1 T2) (arr (T1 (mu T1 T2)) (T2 (mu T1 T2))).
fold : {T1}{T2} subtp (arr (T1 (mu T1 T2)) (T2 (mu T1 T2))) (mu T1 T2).
inf/arr : subtp T1 S1 -> subtp S2 T2 -> subtpinf (arr S1 S2) (arr T1 T2).
\end{verbatim}
\end{minipage}
\caption{An Encoding of Subtyping in $\CoLFCircular$}
\label{fig:encoding_of_subtyping_in_colf}
\end{figure}
\begin{theorem}
    [Adequacy of Encoding]
    \begin{enumerate}
        \item 
    There is a compositional bijection between recursive types and valid canonical terms of type \verb$tp$
        \item 
        For types $\sigma$ and $\tau$, 
    there is a compositional bijection between valid cyclic subtyping derivations of $\sigma \le \tau$, 
    and valid canonical terms of type \verb$subtp $$\enc{\sigma}\, \enc{\tau}$.
    \end{enumerate}
\end{theorem}

\begin{proof}
    \begin{enumerate}
        \item Directly by induction on the structure of recursive types in the forward direction, 
        and by induction on the structure of the typing derivation in the reverse direction.
        \item By induction on the syntax of the circular derivations in the forward direction, and 
        by induction on the syntax of the higher-order rational terms in the reverse direction.
        Note that cycles in the circular derivations correspond directly to occurrences of recursion constants.
        The validity condition of mixed induction and coinduction coincides with CoLF validity.
    \end{enumerate}
\end{proof}

We give an example of the subtyping derivation of $\mu X. X \arrto X \le \mu X. (X \arrto \bot ) \arrto \top $.
Let $S = \mu X. X \arrto X $ and $T = \mu X. (X \arrto \bot ) \arrto \top $. 
\[
\infer[\m{trans}]{
    (\mt{s\_sub\_t})\ S \le T
}{
    \infer[\m{unfold}]{
        S \le S \arrto S
    }{
    }
    &
    \hspace{-95pt}
    \infer[\m{trans}]{
        S \arrto S \le T
    }{
        \infer[\m{\arrto}]{
            S \arrto S \le (T \arrto \bot) \arrto \top
        }{
            \infer[\m{trans}]{
                T \arrto \bot \le S
            }{
                \infer[\m{\arrto}]{
                    T \arrto \bot \le S \arrto S
                }{
                    \deduce{
                        S \le T
                    }{
                        (\mt{s\_sub\_t})
                    }
                    &
                \hspace{-2pt}
                    \infer[\bot]{
                        \bot \le S
                    }{}
                }
                &
                \hspace{-2pt}
                \infer[\m{fold}]{
                    S \arrto S \le S
                }{
                }
            }
            &
            \hspace{-2pt}
            \infer[\top]{
                S \le \top
            }{
            }
        }
        &
        \hspace{-2pt}
        \infer[\m{fold}]{
            (T \arrto \bot) \arrto \top \le T
        }{
        }
    }
}
\]

Here is the encoding in $\CoLFCircular$: 

\begin{verbatim}
s : tp = mu ([x] x) ([x] x).
t : tp = mu ([x] arr x bot) ([x] top).
s_sub_t : subtp s t = 
    trans (unfold ([x] x) ([x] x)) (trans (subtp/arr (inf/arr
                    (trans (subtp/arr (inf/arr s_sub_t subtp/bot))
                        (fold ([x] x) ([x] x))) subtp/top)) 
                        (fold ([x] arr x bot) ([x] top))).
\end{verbatim}

We note that the circular definition is valid by the presence of the 
constructor \verb#inf/arr# along the trace from \verb#s_sub_t# to itself. 
The presence of the coinductive $\m{arr}$ rule is the validity condition of mixed inductive and coinductive 
definitions. 

There are two key differences between a $\CoLFCircular$ encoding and an Agda encoding.
First, in Agda one needs to use explicit names for $\mu$-bound variables or de Bruijn indices, while 
in CoLF one uses abstract binding trees. Second, Agda does not have built-in coinductive equality but 
CoLF has built-in equality. In Agda, the one step of unfolding \verb$s_sub_t$ is not equal 
to \verb$s_sub_t$, but in $\CoLFCircular$, they are equal.

\newcommand{\upshift}{{\uparrow}}
\newcommand{\downshift}{{\downarrow}}
\newcommand{\ampersand}{\operatorname{\&}}

\newcommand{\tpempty}{\m{empty}}
\newcommand{\tpfull}{\m{full}}
\newcommand{\EMP}{\m{EMP}}
\newcommand{\FULL}{\m{FULL}}
\newcommand{\SUB}{\m{SUB}}

\subsection{Encoding a Polarized Circular Subtyping System for Equirecursive Types}

We present an encoding of a variant Lakhani et al.'s polarized subtyping system \cite{Lakhani22esop} into $\CoLFCircular$. 
The system is circular. 
Due to space constraints, we only present the encoding for the positive types fragment and their emptiness derivations. 
This is an important part in the subtyping system because an empty type is a subtype of any other type.
The full encoding of the polarized subtyping system can be found in 
Appendix~\ref{appendix:encoding_polarized_subtyping}\oftheextendedversion.


\subsubsection{Encoding of Positive Equirecursive Types.}

 The equirecursive nature is captured 
by a signature $\Sigma$ providing recursive definitions for type names $t^+$.

\begin{center}
\begin{tabular}{lll}
    $\tau^+, \sigma^+ $ & $ ::= $ & $t_1^+ \otimes t_2^+ \mid \mathbf{1} \mid t_1^+ \oplus t_2^+ \mid \mathbf{0} $ \\
    $\Sigma $ & $ ::= $ & $\cdot \mid \Sigma, t^+ = \tau^+ $ \\
\end{tabular}
\end{center}

Equirecursive types are directly encoded as recursion constants in the system, 
and the framework automatically provides equirecursive type equality checking.
Because equirecursive types are circular, positive types are encoded as $\cotype$.

\begin{nolinenumbers}
\begin{multicols}{2}
\begin{verbatim}
postp : cotype.

times : postp -> postp -> postp.
one : postp.
plus : postp -> postp -> postp.
zero : postp.
\end{verbatim}
\end{multicols}
\end{nolinenumbers}

\begin{theorem}[Adequacy of Type Encoding]
    There is a bijection between circular types defined in an object signature for the 
    positive types fragment and canonical forms of the \verb#postp# 
    in $\CoLFCircular$.
\end{theorem}
\begin{proof}
    By induction on the syntax in both directions.
\end{proof}

\subsubsection{Encoding of the Emptiness Judgment.}
The emptiness judgment $t \tpempty$ is defined by the following rules.
We stress that these rules are to be interpreted coinductively.

\begin{mathpar}
    \inferrule{
    }{
        0 \tpempty
    }(\mathbf{0} \EMP)

    \inferrule{
        t = t_1 \oplus t_2 \in \Sigma
        \\
        t_1 \tpempty
        \\
        t_2 \tpempty
    }{
        t \tpempty
    }(\oplus \EMP)

    \inferrule{
        t = t_1 \otimes t_2 \in \Sigma
        \\
        t_1 \tpempty
    }{
        t \tpempty
    }(\otimes \EMP_1)
    \hspace{1em}
    \inferrule{
        t = t_1 \otimes t_2 \in \Sigma
        \\
        t_2 \tpempty
    }{
        t \tpempty
    }(\otimes\EMP_2)
\end{mathpar}

In $\CoLFCircular$, the rules are encoded as follows.
The coinductive nature is reflected by the typing of \verb#empty : postp -> cotype#, which 
postulates that the predicate \verb#empty# is to be interpreted coinductively.

\begin{verbatim}
empty : postp -> cotype.
zero_emp : empty zero.
plus_emp : empty T1 -> empty T2 -> empty (plus T1 T2).
times_emp_1 : empty T1 -> empty (times T1 T2).
times_emp_2 : empty T2 -> empty (times T1 T2).
\end{verbatim}

\begin{theorem}[Adequacy of Encoding]
    There is a bijection between the circular derivations of $t \tpempty$ and 
    the canonical forms of the type $\mathtt{empty\,\enc{t}}$.
\end{theorem}
\begin{proof}
    By induction on the syntax of the circular derivation in both directions.
\end{proof}

As an example, we may show that the type $t$, where $t = \mathbf{1} \otimes t$, is empty by the following circular derivation.
\[
    \infer[\otimes \EMP_2]{
        (\mt{t\_empty})\ \mathbf{1} \otimes t\tpempty
    }{
        (\mt{t\_empty})\ t \tpempty
    }\]


This derivation can be encoded as follows.

\begin{verbatim}
t : postp = times one t.
t_empty : empty t = times_emp_2 t_empty.
\end{verbatim}

The reader is advised to take a look at 
Appendix~\ref{subsection:examples_of_polarized_subtyping_proof}\oftheextendedversion[ ]for two simple yet elegant examples of subtyping derivations.

\section{Related Work}

\textit{Cyclic $\lambda$-Calculus and Circular Terms.}
Ariola and Blom \cite{Ariola97tacs}, and Ariola and Klop \cite{Ariola97ic} studied the confluence property of reduction of cyclic $\lambda$-calulus.
Their calculus differs from CoLF in several aspects. Their calculus is designed to capture reasoning principles 
of recursive functions and thus has a general recursive let structure that 
can be attached to terms at any levels. Terms are equated up to infinite Lévy-Longo trees (with decidable equality), but equality as B\"ohm trees 
is not decidable.
CoLF is designed for circular terms and circular derivations, and all recursive definitions occur 
at the top level. Terms are equated up to infinite B\"ohm trees and the equality is decidable.
Our equality algorithm is adapted from
Huet'algorithm for the regular B{\"{o}}hm trees \cite{Huet98mscs}.
Equality on first-order terms has been studied both in its own respect \cite{Endrullis11tcs} and 
in the context of subtyping for recursive types \cite{Amadio93toplas,Brandt98fi,Danielsson10mpc,Ligatti17toplas}.
Our algorithm when applied to first-order terms is ``the same''.
Courcelle \cite{Courcelle83tcs} and Djelloul et al. \cite{Djelloul08tplp} have studied the 
properties of first-order circular terms.
 Simon \cite{Simon06phd} designed 
a coinductive logic programming language based on the first-order circular terms. Contrary to $\CoLFCircular$, 
there are no mutual dependencies between inductive 
and coinductive predicates in Simon's language.

\textit{Logical Frameworks.}
Harper et al. \cite{Harper93jacm} designed the logical framework LF, which this work extends upon.
 Pfenning et al. later adds notational definitions \cite{Pfenning98pstt}. 
The method of hereditary substitution was developed as part of the research 
on linear and concurrent logical frameworks \cite{Cervesato96lics,Watkins02tr,Cervesato02tr}.
Harper and Licata demonstrated the method in formalizing the metatheory of simply typed $\lambda$-calculus \cite{Harper07jfp}. 
In his master's thesis, Chen has investigated a mixed inductive and coinductive logical framework with an
infinite stack of priorities but only in the context of a first-order type theory \cite{Chen21ms}.

\textit{Mixed Induction and Coinduction and Circular Proof Systems.}
The equality and subtyping systems of recursive types \cite{Amadio93toplas,Brandt98fi,Danielsson10mpc,Ligatti17toplas,Lakhani22esop}
have traditionally recognized coinduction and more recently mixed induction and coinduction as 
an underlying framework. 
Fortier and Santocanale \cite{Fortier13csl} devised a circular proof system for propositional linear sequent calculus
with mixed inductive and coinductive predicates. This system together 
with Charatonik et al.'s Horn $\mu$-calculus \cite{Charatonik98lics} motivated the validity condition of CoLF.
Brotherston and Simpson devised an infinitary and a circular proof system as methods of carrying out induction \cite{Brotherston05tableaux,Brotherston11jlc}.
Due to the complexity of their validity condition, the encoding of Brotherston and Simpson's system in full generality and Fortier and Santocanale's system is currently not immediate and is considered in ongoing work.


\section{Conclusion}

We have presented the type theory of a novel logical framework 
with higher-order rational terms,
that admit coinductive and mixed inductive and coinductive interpretations.
We have proposed the prepattern variables and prepattern $\Pi$-types 
to give a type-theoretic formulation of regular B\"ohm trees. 
Circular objects and derivations are represented as higher-order rational terms, 
as demonstrated in the case study of the subtyping deductive systems for recursive types.

We once again highlight the methodology of logical frameworks and what $\CoLFCircular$ accomplishes. 
Logical frameworks internalize equalities that are present in the term model 
for an object logic. LF \cite{Harper93jacm} internalizes $\alpha\beta\eta$-equivalence of 
the dependently typed $\lambda$-calculus. Within LF, 
one is not able to write a specification that distinguishes two terms that are $\alpha$ or $\beta$-equivalent, 
because those two corresponding derivations are identical in the object logic. 
Similarly, the concurrent logical framework CLF \cite{Watkins02tr} internalizes equalities
of concurrent processes that only differ in the order of independent events. 
The logical framework $\CoLFCircular$ 
internalizes the equality of circular derivations. Using $\CoLFCircular$, one cannot 
write a specification that distinguishes between two different finitary representations of the same circular proof. 
It is this property that makes $\CoLFCircular$ a more suitable framework for encoding 
circular derivations than existing finitary frameworks.

\par\noindent\textbf{Acknowledgments.} 
We would like to thank Robert Harper and Brigitte Pientka for
insightful discussion on the research presented here and the anonymous reviewers
for their helpful comments and suggestions.

\bibliographystyle{splncs04}
\bibliography{citationsconf}


\include{appendix}

\end{document}

%% file: macros.tex
\usepackage[utf8]{inputenc}
\usepackage{amsmath}
\usepackage{enumerate}
\usepackage{mathpartir}
\usepackage{mathtools}
\usepackage{pmboxdraw}
\usepackage[T1]{fontenc}

\usepackage{ amssymb }
\usepackage[normalem]{ulem}
\usepackage{mathpartir}
\usepackage{proof}
\usepackage{mathtools}
\usepackage{longtable}
\usepackage{supertabular}
\usepackage{ stmaryrd }

\usepackage[pass]{geometry}
\usepackage{changepage}
\usepackage{longtable}
\usepackage{comment}
\usepackage{color}   
\usepackage[bookmarks=false]{hyperref}
\hypersetup{
    colorlinks,
    citecolor=black,
    filecolor=black,
    linkcolor={blue!50!black},
    urlcolor=black
}

\usepackage{tikz} 

\usepackage{diagbox}
\usepackage{pdflscape}

\usepackage{makecell}
\usepackage{changepage}
\usepackage{xcolor}
\usepackage{multicol}
\usepackage{multirow}
\usepackage{listings}
\usepackage{etoolbox}
\usepackage{environ}
\usepackage{newunicodechar}
\usepackage{bm}
\usepackage{vwcol}  
\usepackage{lineno}

\usepackage{framed}
\usepackage{cprotect}

\usepackage{orcidlink}

\newunicodechar{∞}{\ensuremath{\mathnormal\infty}}
\newunicodechar{♯}{\ensuremath{\mathnormal\sharp}}
\newunicodechar{⌊}{\ensuremath{\mathnormal{\lfloor}}}
\newunicodechar{⌋}{\ensuremath{\mathnormal{\rfloor}}}
\newunicodechar{≟}{\ensuremath{\mathnormal{\stackrel{?}{=}}}}
\newunicodechar{≡}{\ensuremath{\mathnormal{\equiv}}}

\newcommand\validsig{\mathrel{\operatorname{\mathsf{sig}}}}
\newcommand\validctx{\mathrel{\operatorname{\mathsf{ctx}}}}
\newcommand\cotype{\operatorname{\mathsf{cotype}}}
\newcommand\type{\operatorname{\mathsf{type}}}
\newcommand\typecotype{\operatorname{\mathsf{(co)type}}}
\newcommand\kind{\operatorname{\mathsf{kind}}}

\newcommand{\m}[1]{\operatorname{\mathsf{{#1}}}}

\newcommand{\mb}[1]{\mathop{\operatorname{\mathbf{{#1}}}}}
\newcommand{\mt}[1]{\mathop{\operatorname{\mathtt{{#1}}}}}

\newcommand{\pitype}[3]{\ensuremath{\Pi {#1} : {#2} .\, {#3}}}

\newcommand\hra\hookrightarrow

\newcommand{\Piargs}[3]{\Pi {#1}:{#2}.\,{#3}}

\newcommand{\pdepth}[1]{_{({#1})}}
\newcommand{\depth}[1]{\pdepth{#1}}
\newcommand\uperased[1]{^{{#1}^o}}
\newcommand\erased[1]{{{#1}^o}}

\newcommand{\bN}{\mathbb{N}}

\newcommand{\enc}[1]{\ulcorner {#1}\urcorner}

\newcounter{numbered}
\newenvironment{numbered}{
    \setcounter{numbered}{0}
    }{
    \setcounter{numbered}{0}
    }

\newcommand{\previousNumber}{\the\numexpr\value{numbered}-1\relax }

\newcommand{\red}[1]{\textcolor{red}{#1}}

\newcommand{\CoLFCircular}{\mathnormal{\text{CoLF}}}

\newcommand{\ctxeqassump}[3]{({#1}\vdash{#2}={#3})}

\newcommand{\varrenamesubst}[1]{\llbracket {#1} \rrbracket}

\newcommand{\guardedin}{\rtimes{}}

\newcommand{\snil}{\operatorname{\mathsf{()}}}
\newcommand{\prepatof}{\mathrel{\hat{:}}}
\newcommand{\prepatofoption}{\mathrel{\overset{\resizebox{.5em}{.18em}{$\bm{(\wedge)}$}}{:}}}
\newcommand{\isprepattern}{\operatorname{\mathsf{prepat}}}
\newcommand{\iscontractive}{\operatorname{\mathsf{contra}}}
\newcommand{\isvalidtrace}{\operatorname{\mathsf{validtrace}}}
\newcommand{\prepatpi}[3]{\Pi {#1} \prepatof {#2}. \, {#3}}

\newcommand{\rsignature}{{\Xi; \Sigma}}

\newcommand{\prepatternvar}[1]{[{#1}]}

\newcommand{\vs}{\vdash_\Sigma}

\newcommand{\citeyear}{\cite}

%% file: appendix.tex
\appendix
\title{Appendix}
\author{}
\institute{}
\maketitle
\section{Hereditary Substitution and Renaming Substitution}
\label{appendix:hereditary_substitution_and_renaming_substitution}
\begin{longtable}[l]{llllll}
    \boxed{A^o = \tau}  \\
    $(\pitype{x}{A_2}{A_1})^o = (A_2^o) \to (A_1^o)$ \\
    $(\prepatpi{x}{A_2}{A_1})^o = * \to (A_1^o)$ \\
    $(P)^o = *$ \\
    \boxed{[N/x]^\tau K = K'} \\
    $[N/x]^\tau\type  = \type $ \\
    $[N/x]^\tau \cotype =\cotype$ \\
    $[N/x]^\tau(\Pi y : A.\ K) = \Pi y: [N/x]^\tau A.\, [N/x]^\tau K$ & $y \ne x$\\
    \boxed{[N/x]^\tau A = A'} \\
    $[N/x]^\tau P = [N/x]^\tau P$ \\
    $[N/x]^\tau(\Pi y : A_2.\ A_1) = \Pi y : [N/x]^\tau A_2. [N/x]^\tau A_1$ & $y \ne x$\\
    \boxed{[N/x]^\tau P = P'} \\
    $[N/x]^\tau(a \cdot S) = a \cdot ([N/x]^\tau S)$ \\
    \boxed{[N/x]^\tau M = M'} \\
    $[N/x]^\tau R = [N/x]^\tau R$ \\
    $[N/x]^\tau(\lambda y.M) = \lambda y. [N/x]^\tau M$  & $y \ne x $ \\
    \boxed{[N/x]^\tau R = R'} \\
    $[N/x]^\tau( x \cdot S) = ([N/x]^\tau S) \rhd^{\tau} N$ \\
    $[N/x]^\tau( y \cdot S) =  y \cdot ([N/x]^\tau S) $ & $ y \ne x$\\
    $[N/x]^\tau( c \cdot S) = c \cdot ([N/x]^\tau S) $\\
    $[N/x]^\tau( r \cdot S) = r \cdot ([N/x]^\tau S) $\\
    \boxed{[N/x]^\tau  S = S'} \\
    $[N/x]^\tau\snil = \snil$ \\
    $[N/x]^\tau(M;S) = ([N/x]^\tau M); ([N/x]^\tau S)$ \\
    $[N/x]^\tau(\prepatternvar x;S) = \m{undefined}$  \\
    $[N/x]^\tau(\prepatternvar z;S) = \prepatternvar z; ([N/x]^\tau S)$ & $x \ne z$ \\
    \boxed{S \rhd ^\tau N = R'} \\
    $\snil \rhd ^* R = R$ \\
    $(N; S) \rhd^{\tau_2 \to \tau_1} (\lambda x. M) = S \rhd^{\tau_1} ([N/x]^{\tau_2}M)$ \\
    $(\prepatternvar y; S) \rhd^{* \to \tau_1} (\lambda x. M) = S \rhd^{\tau_1} (\varrenamesubst{y/x}M)$ \\
    \boxed{[N/x]^\tau \Gamma =  \Gamma'} \\
    $[N/x]^\tau\cdot = \cdot$ \\
    $[N/x]^\tau(\Gamma, y : A) = ([N/x]^\tau\Gamma),  y : ([N/x]^\tau A)$ \\
    $[N/x]^\tau(\Gamma, y \prepatof A) = ([N/x]^\tau\Gamma),  y \prepatof ([N/x]^\tau A)$ \\
    \boxed{\varrenamesubst{y/x}K = K'} \\
    $\varrenamesubst{y/x}{\type} = \type$ \\
    $\varrenamesubst{y/x}{\cotype} = \cotype$ \\
    $\varrenamesubst{y/x}{\Pi z : A.\, K} = \Pi z : \varrenamesubst{y/x}{A}. \, \varrenamesubst{y/x}K$  & $z \ne x, y$\\
    $\varrenamesubst{y/x}{\Pi z \prepatof A.\, K} = \Pi z \prepatof \varrenamesubst{y/x}{A}. \, \varrenamesubst{y/x}K$ & $z \ne x, y$\\
    \boxed{\varrenamesubst{y/x}A = A'} \\
    $\varrenamesubst{y/x}{\Pi z : A_2. A_1} = \Pi z : \varrenamesubst{y/x}{A_2}. \, \varrenamesubst{y/x}A_1$ & $z \ne x, y$\\
    $\varrenamesubst{y/x}{\Pi z \prepatof A_2. A_1} = \Pi z \prepatof \varrenamesubst{y/x}{A_2}. \, \varrenamesubst{y/x}A_1$ & $z \ne x, y$\\
    $\varrenamesubst{y/x}{P} = \varrenamesubst{y/x}{P} $ \\
    \boxed{\varrenamesubst{y/x}P = P'} \\
    $\varrenamesubst{y/x}{a} = a $\\
    $\varrenamesubst{y/x}{P\, M} = \varrenamesubst{y/x}{P} \, \varrenamesubst{y/x}{M}$ \\
    \boxed{\varrenamesubst{y/x}M = M'} \\
    $\varrenamesubst{y/x}{R} = \varrenamesubst{y/x}{R} $\\
    $\varrenamesubst{y/x}{\lambda z.\, M} = \lambda z. \, \varrenamesubst{y/x}{M}$& $z \ne x, y$ \\
    \boxed{\varrenamesubst{y/x}R = R'} \\
    $\varrenamesubst{y/x}{x \cdot S} = y\cdot \varrenamesubst{y/x}{S} $\\
    $\varrenamesubst{y/x}{z \cdot S} = z\cdot \varrenamesubst{y/x}{S} $ & $z \ne x$\\
    $\varrenamesubst{y/x}{c \cdot S} = c\cdot \varrenamesubst{y/x}{S} $\\
    $\varrenamesubst{y/x}{r \cdot S} = r\cdot \varrenamesubst{y/x}{S} $\\
    \boxed{\varrenamesubst{y/x}{S} ={S'}} \\
    $\varrenamesubst{y/x}{\snil} = \snil$ \\
    $\varrenamesubst{y/x}{(M ;  S)} = \varrenamesubst{y/x}{M} ; (\varrenamesubst{y/x}{ S})$ \\
    $\varrenamesubst{y/x}{(\prepatternvar x ;  S)} = \prepatternvar y ; (\varrenamesubst{y/x}{ S})$ \\
    $\varrenamesubst{y/x}{(\prepatternvar z ;  S)} = \prepatternvar z ; (\varrenamesubst{y/x}{ S})$ \\
    \boxed{\varrenamesubst{y/x} \Gamma =  \Gamma'} \\
    $\varrenamesubst{y/x}\cdot = \cdot$ \\
    $\varrenamesubst{y/x}(\Gamma, z : A) = (\varrenamesubst{y/x}\Gamma),  z : (\varrenamesubst{y/x} A)$& $z \ne x, y$ \\
    $\varrenamesubst{y/x}(\Gamma, z \prepatof A) = (\varrenamesubst{y/x}\Gamma),  z \prepatof (\varrenamesubst{y/x} A)$& $z \ne x, y$ \\
    \end{longtable}

\section{Omitted Rules}
Below are some straightforward rules that are omitted from the main text due to the page limit.
\subsection{Prepattern Checking}
\begin{multicols}{2}
    \boxed{A \isprepattern}
    \begin{mathpar}
        \inferrule{
        }{
            P \isprepattern
        }
    
        \inferrule{
            A_1 \isprepattern
        }{
            \prepatpi{x}{A_2}{A_1}
        }
    \end{mathpar}
    
    \boxed{S \isprepattern}
    \begin{mathpar}
        \inferrule{
        }{
            \snil \isprepattern
        }
    
        \inferrule{
            S \isprepattern
        }{
            \prepatternvar x; S \isprepattern
        }
    \end{mathpar}
    \end{multicols}

\subsection{Guardedness Checking}

\label{appendix:guardedness_checking}
We devise an algorithm for checking the guardedness of recursive definitions.
Let $C$ denote a set of constructors, it is easy to check whether there exists a coinductive constructor of the highest priority.
We use judgment $C \isvalidtrace$ to denote this check and omit the rules.
We write $Q$ for sets of recursion constants, and define the judgment $ Q; C\vdash _\Sigma r \guardedin M $ to mean that for all occurrences 
of $r$ in $M$ is properly guarded by some coinductive constructor, where $Q$ holds recursion constants that we have explored and $C$ holds constructors
we have encountered.
The rules for deriving this judgment is syntax directed on $M$ and are shown in Fig.~\ref{fig:guardedness_circular_def}.
Note that this judgment does not track the free variables in $M$ so $M$ may be open.
Thus, for open $M$, if $\{\}; \{\}\vdash _\Sigma r \guardedin  M $, $r$ is guarded in any closed (hereditary) substitution instances of $M$.
The judgment is parametrized by a signature $\Sigma$, as we need to have access to the definition for any recursion constant in $M$.
To ensure the validity of any term occurring in the signature, it suffices to check that for all 
recursive definitions $r : A = M$, $r$ must be guarded in $M$.
We define an auxiliary judgment $Q; C\vdash_\Sigma r \guardedin S $ to mean that 
$r$ is guarded in each $M \in S$.

\begin{figure}[h]
\boxed{{Q; C} \vs r \guardedin M  }
\begin{mathparpagebreakable}
    \inferrule{
        Q; C \vdash _ \Sigma r \guardedin M 
    }{
        Q; C \vs r \guardedin \lambda x.\, M 
    }

    \inferrule{
        Q; C\cup \{c\} \vs r \guardedin S
    }{
        Q; C \vs r \guardedin c \cdot S 
    }

    \inferrule{
       Q; C \vs r \guardedin S
    }{
        Q; C \vs r \guardedin x \cdot S
    } 

    \inferrule{
        r \ne r' \\
        r' \in Q
    }{
        Q; C \vs r \guardedin r' \cdot S 
    }

    \inferrule{
        C \isvalidtrace
    }{
        Q ; C \vs r \guardedin r \cdot S 
    }

    \inferrule{
        r' : A = M \in \Sigma
        \\
         S \isprepattern
        \\
        Q \cup \{r'\}; C \vdash _\Sigma r \guardedin M
    }{
        Q; C\vs r  \guardedin r' \cdot S
    } (r\ne r', \text{$r \notin Q$}) (*) \\
\end{mathparpagebreakable}
\boxed{Q; C\vs r\guardedin S} 
\begin{mathpar}
    \inferrule{
    }{
        Q; C\vs r \guardedin \snil 
    }

    \inferrule{
        Q; C\vs r \guardedin M 
        \\
        Q; C\vs r \guardedin S
    }{
        Q; C\vs r \guardedin M; S
    }

    \inferrule{
        Q; C\vs r \guardedin S 
    }{
        Q; C\vs r \guardedin \prepatternvar x; S 
    }
\end{mathpar}
\caption{Guardedness Checking}
\label{fig:guardedness_circular_def}
\end{figure}

\begin{theorem}
    [Decidability of guardedness checking]
    It is decidable given $\Sigma$ whether $Q;C \vdash_\Sigma {r} \guardedin M $ given arbitrary well-formed $r$, $M$, $Q$ and $C$.
\end{theorem}
\begin{proof}
    The only rule that does not analyze the structure of the term is the rule $(*)$. 
    It is impossible that the proof search for guardedness invokes this rule infinitely many times, because the rule $(*)$ strictly increases the size of $Q$ from bottom to top, but 
     there can only be finitely many distinct recursion constants in a signature.
\end{proof}

\section{Metatheorems of  Equality Checking}
\label{appendix:metatheorems_for_equality_checking}

\begin{theorem}
    $\Delta; \Theta \vdash _\Sigma - = - $ is an equivalence relation (i.e., reflexive, symmetric and transitive).
 \end{theorem}
 \begin{proof}
    Straightforward induction for reflexivity and symmetry.
    Transitivity can be proved by merging two equality proofs, replacing rules (1) with (2) and (3) when the bisimulation cannot be constructed.
    Transitivity can also be proved by appealing to soundness and then completeness, i.e., all three terms expand to the same B\"ohm tree.
 \end{proof}
 
\begin{theorem}
    [Compatibility] 
    If  $\Theta' \vdash N = N'$ and $\Theta \vdash M = M'$, then for all $\tau$, $\Theta \cup \Theta' \vdash [N/x]^\tau M = [N'/x]^\tau M'$ 
    if both substitutions are defined.
\end{theorem}
\begin{proof}

    We have the following steps. 
    \begin{enumerate}
        \item 
    By the soundness,  $exp(N) = exp(N')$ and $\exp(M) = \exp(M')$. 
        \item 
    Then, for all $\tau$, $[\exp(N)/x]^\tau \exp(M)= [\exp(N')/x]^\tau  \exp(M')$ if defined.
        \item 
    By commutation, $\exp([N/x]^\tau M) = \exp([N'/x]^\tau M')$.  
        \item 
    By completeness, $\Theta \cup \Theta' \vdash [N/x]^\tau M = [N'/x]^\tau M'$.
    \end{enumerate}
\end{proof}


\section{Estimating the Maximum Number of Equations}

\label{appendix:equality_upper_bound_estimate}

A very rough upper bound can be estimated for the equality algorithm. Let $b$ be the maximum breadth of all terms in the signature,  
$d$ be the maximum depth, and $l$ denotes the maximum length of abstractions (determined by a type). 
The structure of a term is completely determined by its trace.
The number of traces $p$ in a term of maximum depth $d$ and breadth $b$ can be estimated by
$p = \Sigma_{i=1}^db^{i-1}= \frac{1-b^d}{1-b}$. For each trace, there can be at most $l$ binders. 
So the maximum number of variables that can possibly appear in a term is $(l+1)*p$. 
For each position, we could have constants, recursion constants, variables, or empty position, over counting 
the occurrences of constants in binder positions. 
Let $n$ and $m$ denotes the number of constants and recursion constants in the signature.
Thus, a rough upper bound for the number of 
terms of finite depth and breadth, and finite abstraction length, is $((l+1)*p)^{1 + m + n + (l+1) *p }$. 

We note that this is a very rough upper bound and in practice the actual number of assumptions will be much smaller.
Indeed, one optimization that is performed in the implementation is to first flatten the recursive definition \cite{Lakhani22esop}. 
Flattening  reduces maximum depth of all terms to 2 and thus avoids the exponential blowup in factor $p$. 
In any case, the rough upper bound suffices to show that the algorithm is decidable.

\section{Expansion of Higher-order Rational Terms as B\"ohm trees}
The function $\exp\depth k(M)$ denotes expanding term $M$ into a B\"ohm tree of depth $k$.
The infinite unfolding of $M$ as a B\"ohm tree is the limit of all the finite approximates $\exp\depth k(M)$ \cite{Barendregt85book}.

\label{appendix:expansion_as_bohm_trees}
    \begin{flushleft}
    \begin{tabular} {lll}
        \boxed{\exp\pdepth k(M)=\pdepth{k}M'}
        \\
        $\exp\pdepth {0}(
            M
        ) $ & $ = \pdepth 0 \bot$ \\
        $\exp\pdepth {k+1}(
            \lambda x.\, M
        ) $ & $ = \pdepth{k+1}$ & $ \lambda x. \exp\pdepth{k+1}(M)$ \\
        $\exp\pdepth {k+1}(
            R
        ) $ & $ = \pdepth{k+1}$ & $ \exp(R)$ \\
    \end{tabular}
    \begin{tabular}{lll}
        \boxed{\exp\pdepth k(R)=\pdepth{k}R'}
        \\
        $\exp\pdepth {0}(
            R
        ) $ & $ = \pdepth 0 \bot$ \\
        $\exp\pdepth {k+1}(
            c \cdot S
        ) $ & $ =\pdepth{k+1} $ & $c\cdot (\exp\pdepth{k}(S))$\\
        $\exp\pdepth {k+1}(
            x \cdot S
        ) $ & $ = \pdepth{k+1}$ & $x\cdot (\exp\pdepth{k}(S))$\\
        $\exp\pdepth {k+1}(
           r \cdot S
        ) $ & $ = \pdepth{k+1}$ & $ \exp\pdepth{k+1}(S \rhd^{A^o} M)$ if $r : A = M \in \Sigma$ \\
    \end{tabular}
    \begin{tabular}{lll}
        \boxed{\exp\pdepth k(S)=\pdepth{k}S'}
        \\
        $\exp\pdepth {0}(
            S
        ) $ & $ = \pdepth 0 \bot$ \\
        $\exp\pdepth {k+1}
            \snil
         $ & $ =\pdepth{k+1}$ & $ \snil$\\
        $\exp\pdepth {k+1}(
            M ; S
        ) $ & $ = \pdepth{k+1}$ & $(\exp\pdepth{k+1}(M)); (\exp\pdepth{k+1}(S))$\\
        $\exp\pdepth {k+1}(
          \prepatternvar  x; S
        ) $ & $ = \pdepth{k+1}$ & $\prepatternvar x; (\exp\pdepth{k+1}(S))$ \\
        \end{tabular}
    \end{flushleft}

\section{Type Checking Rules for $\CoLFCircular$}
\label{appendix:type_checking_rules_for_colf}

\subsection{Presuppositions}
The judgment $\vdash_\rsignature \Gamma \validctx$ presupposes $\vdash_\Sigma \Xi \validsig$. 
All judgments of the form $\Gamma \vdash _\rsignature \mathcal J$ presuppose $\vdash_\rsignature \Gamma \validctx$. 
The judgment  $\Gamma \vdash_\rsignature S \rhd K \Rightarrow K'$ presupposes $\Gamma\vdash _\rsignature K \Leftarrow \kind$.
The judgments $\Gamma \vdash_\rsignature M \Leftarrow A$ and  $\Gamma \vdash_\rsignature S \rhd A \Rightarrow P$  presuppose 
$\Gamma\vdash _\rsignature A \Leftarrow \typecotype$.

\subsection{Rules}
We write $|\Gamma|$ for the list of variables $\Theta$ in $\Gamma$. For instance, if $\Gamma = x : A, y : B$, then $\Theta = x, y$ and we write $|x : A, y : B| = x, y$. The notion is useful 
because term equality algorithm needs to know the free variables in the term but not their types.

    $\boxed{\Sigma \validsig}$ 
    \begin{mathparpagebreakable}
            \inferrule{
                \vdash_\Sigma \Sigma \validsig 
            }{
                \Sigma \validsig
            }
    \end{mathparpagebreakable}
    
    $\boxed{\vdash_{\Sigma}\Xi \validsig}$ 
    \begin{mathparpagebreakable}
            \inferrule{
            }{
             \vdash_{\Sigma}    \cdot \validsig
            }
        
            \inferrule{
               \vdash_{\Sigma} \Xi \validsig 
                \\ 
                \vdash_{\rsignature} K \Leftarrow \kind
            }{
               \vdash_{\Sigma} \Xi, a : K \validsig
            }
        
            \inferrule{
                \vdash_{\Sigma}\Xi \validsig 
                \\ 
                \vdash_{\Sigma} A \Leftarrow \typecotype
            }{
                \vdash_{\Sigma}\Xi, c : A \validsig
            }
    
            \inferrule{
                \vdash_{\Sigma}\Xi \validsig 
                \\ 
                \vdash_{\rsignature} A \Leftarrow \typecotype
                \\
                [{
                \vdash_{\rsignature} M \Leftarrow A
                }]^{\m{1:deferred}}
                \\
                A \isprepattern
                \\
                M \iscontractive
                \\
                \{\};\{\}\vdash _\Sigma r \guardedin M 
            }{
                \vdash_{\Sigma}\Xi, r : A = M \validsig
            }
        \end{mathparpagebreakable}
        
        $\boxed{\vdash_\rsignature\Gamma \validctx}$
    \begin{mathpar}
            \inferrule{
            }{
                \vdash_\rsignature \cdot \validctx
            }
        
            \inferrule{
                \vdash_\rsignature \Gamma \validctx 
                \\ 
                \Gamma\vdash_\rsignature A \Leftarrow \typecotype
            }{
                \vdash_\rsignature \Gamma, x : A \validctx
            }
    
            \inferrule{
                \vdash_\rsignature \Gamma \validctx 
                \\ 
                \Gamma\vdash_\rsignature A \Leftarrow \typecotype
            }{
                \vdash_\rsignature \Gamma, x \prepatof A \validctx
            }
        \end{mathpar}
    
        $\boxed{\Gamma \vdash_\rsignature K \Leftarrow \kind}$
        \begin{mathpar}
            \inferrule{
            }{
                \Gamma \vdash_\rsignature \type \Leftarrow \kind
            }
        
            \inferrule{
            }{
                \Gamma \vdash_\rsignature \cotype \Leftarrow \kind
            }
        
            \inferrule{
                \Gamma \vdash_\rsignature A \Leftarrow \typecotype
                \\
                \Gamma, x : A\vdash_\rsignature K \Leftarrow \kind
            }{
                \Gamma \vdash_\rsignature  \Pi x : A . \, K \Leftarrow \kind
            }
    
            \inferrule{
                \Gamma \vdash_\rsignature A \Leftarrow \typecotype
                \\
                \Gamma, x \prepatof A\vdash_\rsignature K \Leftarrow \kind
            }{
                \Gamma \vdash_\rsignature  \Pi x \prepatof A . \, K \Leftarrow \kind
            }
        \end{mathpar}
        
        $\boxed{\Gamma \vdash_\rsignature A \Leftarrow \typecotype}$ 
        \begin{mathpar}
            \inferrule{
                \Gamma \vdash_\rsignature A_2 \Leftarrow \typecotype\\
                \Gamma, x : A_2\vdash_\rsignature A_1 \Leftarrow \typecotype
            }{
                \Gamma \vdash_\rsignature \Pi x : A_2. \, A_1 \Leftarrow \typecotype
            }
    
            \inferrule{
                \Gamma \vdash_\rsignature A_2 \Leftarrow \typecotype\\
                \Gamma, x \prepatof A_2\vdash_\rsignature A_1 \Leftarrow \typecotype
            }{
                \Gamma \vdash_\rsignature \Pi x \prepatof A_2. \, A_1 \Leftarrow \typecotype
            }
        
            \inferrule{
                \Gamma \vdash_\rsignature P \Rightarrow K \\
                K = \type/\cotype
            }{
                \Gamma \vdash_\rsignature P \Leftarrow \typecotype
            } 
        \end{mathpar}

        \boxed{\Gamma\vdash_\rsignature P \Rightarrow K}
        \begin{mathpar}
            \inferrule{
                a : K \in \Xi
                \\
                \Gamma \vdash _\rsignature S \rhd K \Rightarrow K'
            }{
                \Gamma\vdash_\rsignature a\cdot S \Rightarrow K'
            }
        \end{mathpar}
    
        \boxed{\Gamma \vdash_\rsignature S \rhd K \Rightarrow K'} 
        \begin{mathpar}
        \inferrule{
        }{
            \Gamma  \vdash_\rsignature \snil  \rhd K \Rightarrow K
        }
    
        \inferrule{
            \Gamma \vdash_\rsignature M \Leftarrow A_2 \\
            [M/x]^{\erased{A_2}}K = K'
            \\
            \Gamma \vdash_\rsignature S \rhd K' \Rightarrow K''
        }{
            \Gamma  \vdash_\rsignature M;S  \rhd \Pi x : A_2. \, K \Rightarrow  K''
        }
    
        \inferrule{
            y \prepatof A_2' \in \Gamma\\
            \Gamma \vdash_\rsignature A_2' = A_2
            \\
            \varrenamesubst{y/x}{K} = K'
            \\
            \Gamma \vdash_\rsignature S \rhd K' \Rightarrow K''
        }{
            \Gamma  \vdash_\rsignature\prepatternvar  y;S  \rhd \Pi x \prepatof A_2. \, K \Rightarrow  K''
        }
        \end{mathpar}
    
        \boxed{\Gamma \vdash_\rsignature M \Leftarrow A} 
        \begin{mathpar}
        \inferrule{
            \Gamma  \vdash_\rsignature R \Rightarrow P'
            \\
            \Gamma \vdash_\Sigma P' = P
        }{
            \Gamma  \vdash_\rsignature R \Leftarrow P
        }
    
        \inferrule{
            \Gamma, x: A_2 \vdash_\rsignature M \Leftarrow A_1
        }{
            \Gamma  \vdash_\rsignature \lambda x.\, M \Leftarrow \Pi x:A_2.\,A_1
        }
    
        \inferrule{
            \Gamma, x \prepatof A_2 \vdash_\rsignature M \Leftarrow A_1
        }{
            \Gamma  \vdash_\rsignature \lambda x.\, M \Leftarrow \Pi x \prepatof A_2.\,A_1
        }
        \end{mathpar}
    
        \boxed{\Gamma \vdash_\rsignature R \Rightarrow P} 
        \begin{mathpar}
        \inferrule{
            (x : A \in \Gamma
            \text{ or }
            x \prepatof A \in \Gamma)
            \\
            \Gamma \vdash _\rsignature S \rhd A \Rightarrow P
        }{
            \Gamma  \vdash_\rsignature x  \cdot S \Rightarrow P
        }
    
        \inferrule{
            r : A = M \in \Xi
            \\
            [\text{or $r : A = M \in \Sigma$}] ^ {\m{2:definitions}}
            \\
            \Gamma \vdash _\rsignature  S \rhd A \Rightarrow P
        }{
            \Gamma  \vdash_\rsignature r \cdot  S \Rightarrow P
        }
    
        \inferrule{
            c : A \in \Xi
            \\
            \Gamma \vdash _\rsignature S \rhd A \Rightarrow P
        }{
            \Gamma  \vdash_\rsignature c \cdot S \Rightarrow P
        }
        \end{mathpar}
    
        \boxed{\Gamma \vdash_\rsignature S \rhd A \Rightarrow P} 
        \begin{mathpar}
        \inferrule{
        }{
            \Gamma  \vdash_\rsignature \snil  \rhd P \Rightarrow P
        }
    
        \inferrule{
            \Gamma \vdash_\rsignature M \Leftarrow A_2 \\
            [M/x]^{\erased{A_2}}A_1 = A_1'
            \\
            \Gamma \vdash_\rsignature S \rhd A_1' \Rightarrow P
        }{
            \Gamma  \vdash_\rsignature M;S  \rhd \Pi x : A_2. \, A_1 \Rightarrow  P
        }
    
        \inferrule{
            y \prepatof A_2' \in \Gamma\\
            \Gamma \vdash_\rsignature A_2' = A_2
            \\
            \varrenamesubst{y/x}{A_1} = A_1'
            \\
            \Gamma \vdash_\rsignature S \rhd A_1' \Rightarrow P
        }{
            \Gamma  \vdash_\rsignature \prepatternvar y;S  \rhd \Pi x \prepatof A_2. \, A_1 \Rightarrow  P
        }
        \end{mathpar}
    
        \boxed{ \Gamma \vdash_\Sigma A_1 = A_2 }
        \begin{mathpar}
        \inferrule{
            \Gamma \vdash_\Sigma P_1 = P_2
        }{
            \Gamma \vdash_\Sigma P_1 = P_2
        }
        
        \inferrule{
            \Gamma \vdash_\Sigma A_1 = A_1'
            \\
            \Gamma, x : A_1\vdash_\Sigma A_2 = A_2'
        }{
            \Gamma \vdash_\Sigma \Piargs{x}{A_1}{A_2} =
            \Piargs{x}{A_1'}{A_2'} 
        }
        
        \inferrule{
            \Gamma \vdash_\Sigma A_1 = A_1'
            \\
            \Gamma , x\prepatof A_1\vdash_\Sigma A_2 = A_2'
        }{
            \Gamma \vdash_\Sigma \prepatpi{x}{A_1}{A_2} =
            \prepatpi{x}{A_1'}{A_2'} 
        }
        \end{mathpar}
        
        \boxed{ \Gamma \vdash_\Sigma P_1 = P_2}
        \begin{mathpar}
        \inferrule{
             \cdot ; |\Gamma| \vdash_\Sigma S = S'
        }{
            \Gamma \vdash_\Sigma  a \cdot S = a \cdot S'
        }
        \end{mathpar}
    
\subsection{Metatheorems}

\begin{theorem}
    [Type Checking Respects Argument Restriction]
    Given $\Sigma$ where $\Sigma \validsig$,
    if $\Gamma \vdash_\rsignature M \Leftarrow A$, then for any occurrence of $r\cdot S$ in $M$, 
    $S$ will only be a list of prepattern variables.
\end{theorem}

\begin{proof}
    Directly induction on the typing derivation.
\end{proof}

\begin{theorem}
    [Preservation of Guardedness]
    Given a signature $\Sigma$, and $\Sigma \validsig$,
    if $\Gamma, x : A, \Gamma'\vdash_\rsignature M \Leftarrow B$, 
    $ \vdash_\Sigma r \guardedin M $, 
    $\Gamma \vdash_\rsignature N \Leftarrow A$,
    and $ \vdash_\Sigma r \guardedin N $, 
    then 
    $\vdash_\Sigma r \guardedin [N/x]^{A^o} M $.
\end{theorem}

\begin{proof}
    By induction on the derivation $Q; C\vdash_\Sigma r \guardedin M$.
\end{proof}

\begin{theorem}
    [Compatibility]
    The type equality is a congruence everywhere. 
    \begin{enumerate}
        \item 
    If $\Gamma \vdash_\rsignature M \Leftarrow A_1$ and $\Gamma \vdash _\Sigma A_1 = A_2$, then 
    $\Gamma \vdash_\rsignature M \Leftarrow A_2$. 
        \item
    If $\Gamma, x : A_1, \Gamma' \vdash_\rsignature M \Leftarrow B$, and $\Gamma \vdash _\Sigma A_1 = A_2$, then  
    $\Gamma, x : A_2, \Gamma' \vdash_\rsignature M \Leftarrow B$.
    \end{enumerate}
\end{theorem}
\begin{proof}
    \begin{enumerate}
        \item 
    By induction, invoking the transitivity of equality at the base case.
    \item
    By induction, invoking the symmetry and transitivity of equality at the base case.
    \end{enumerate}
\end{proof}

\section{Concrete Syntax for CoLF}
\label{appendix:concrete_syntax_for_colf}

We adopt the following conventions for concrete syntax throughout the paper:

\begin{enumerate}
    \item 
    Declarations will be written in the \verb#typewriter# font. 
    \item We write usual applications \verb#c M1 M2# instead of the spine form $c\cdot (M_1; M_2)$.
    \item We use curly brackets $\{\}$ for 
    $\Pi$ types, e.g. $\Pi x: A_2. \, K$ and $\Pi x : A_2.\, A_1$ will be written as \verb#{x : A2} K# and \verb#{x : A2} A1#.
    The type $A_2$ may be omitted if it can be inferred, allowing us to just write \verb#{x} K# and \verb#{x} A1#.
    The entire abstraction may be omitted by writing the binder in the capital letter of a type. 
    For example, $ c : \Pi x : A_2. \,\Pi y : A_3.\, a \cdot (x; y)$ maybe written \verb#(c : a X Y)#, which means \verb#(c : {X} {Y} a X Y)#. 
    The capital letters mimic Prolog's style of metavariables. Note that in this case, we do not need to write out the corresponding applications.
    We write \verb#c# instead of \verb#c M1 M2# and the system can infer the actual arguments. 
    We write \verb#A2 -> A1# for $A_2 \to A_1$ (i.e., $\Pi x : A_2.\, A_1$ and $x \notin FV(A_1)$). 
    \item We use square brackets for $\lambda$-abstraction. For example, $\lambda x.\, M$ is written as \verb#([x] M)#. 
    \item We may write underscores in any position to let the system infer the omitted term.
    \item We write \verb#%%# for comments.
\end{enumerate}

\section{Classical Subtyping in Agda vs. $\CoLFCircular$}

The full set of rules is reproduced as follows.

    \begin{mathparpagebreakable}
        \inferrule{
        }{ \bot \le \tau} (\m{bot})
    
        \inferrule{
        }{ \tau\le \top} (\m{top})
        
        \inferrule{
            }{
                \tau \le \tau
        }(\m{refl})
    
        \inferrule{
        }{
            \mu X.\tau_1\to\tau_2
            \le 
            [\mu X.\tau_1 \to\tau_2/X](\tau_1\to\tau_2) 
        }(\m{unfold})
        
        \inferrule{
        }{
            [\mu X.\tau_1\to\tau_2/X](\tau_1\to\tau_2) 
            \le 
            \mu X.\tau_1\to\tau_2
        }(\m{fold})\\
    
        {
        \mprset{fraction={===}}
        \inferrule{
            \tau_1 \le \sigma_1
            \\
            \sigma_2 \le \tau_2
        }{ \sigma_1 \to \sigma_2 \le \tau_1 \to \tau_2} (\m{arr})
        }
    
        \inferrule{
            \tau_1 \le \tau_2
            \\
            \tau_2 \le \tau_3
        }{
            \tau_1 \le \tau_3
        }(\m{trans})
    \end{mathparpagebreakable}

\label{appendix:subtyping_proofs_agda_colf}
The full Agda code is as follows. Note that the last theorem is not true. A cyclic derivation cannot be proved 
definitionally equal to its one-step unfolding automatically.
\begin{verbatim}
{-# OPTIONS --without-K --safe --universe-polymorphism 
            --no-sized-types
            --guardedness --no-subtyping #-}

open import Agda.Builtin.Coinduction
open import Relation.Nullary
open import Agda.Builtin.Equality
open import Data.String 
open import Data.Bool 
open import Relation.Binary
open import Data.Nat 

open import Relation.Nullary.Decidable

data tp : Set where
     bot  : tp
     top  : tp
     _to_  : tp -> tp -> tp
     mu : String -> tp -> tp -> tp
     var : String -> tp

eqstring : String -> String -> Bool
eqstring x y = if (⌊ (x Data.String.≟ y) ⌋) then true else false

subst : tp -> String -> tp -> tp
subst T2 x (var x1) = if ( eqstring x x1) then T2 else (var x1)
subst T2 x (bot ) = bot
subst T2 x (top ) = top
subst T2 x (s1 to s2 ) = (subst T2 x s1) to (subst T2 x s2)
subst T2 x (mu x1 s1 s2) = if (eqstring x x1) then (mu x1 s1 s2)
                    else mu x1 (subst T2 x s1) (subst T2 x s2)

dounfold : tp -> tp
dounfold (mu x1 s1 s2) = ((subst (mu x1 s1 s2) x1 s1) 
                                to (subst (mu x1 s1 s2) x1 s2))
dounfold T = T

data sub : tp -> tp -> Set where
    bot : {T : tp} ->  sub bot T 
    top : {S : tp} ->  sub S top
    s_to : {T1 T2 S1 S2 : tp} -> ∞ (sub T1 S1) -> ∞ (sub S2 T2) 
                        -> sub (S1 to S2) (T1 to T2)
    unfold : {X1 : String}{T1 : tp}{T2 : tp} -> 
                        sub (mu X1 T1 T2) (dounfold (mu X1 T1 T2))
    fold : {X1 : String}{T1 : tp}{T2 : tp} ->  
                        sub (dounfold (mu X1 T1 T2)) (mu X1 T1 T2)
    refl : {T : tp} -> sub T T
    trans : {T1 T2 T3 : tp} -> sub T1 T2 -> sub T2 T3 -> sub T1 T3

s : tp 
s = mu "x" (var "x") (var "x")
t : tp
t = mu "x" ((var "x") to bot) (top)
s_sub_t : sub s t
s_sub_t = trans (unfold) (trans 
    (s_to (♯(trans (s_to (♯ s_sub_t) (♯ bot)) (fold))) (♯ top) )
    fold)

s_sub_t2 : sub s t
s_sub_t2 = trans (unfold) (trans 
    (s_to (♯(trans (s_to (♯ (
 trans (unfold) (trans 
    (s_to (♯(trans (s_to (♯ s_sub_t2) (♯ bot)) (fold))) (♯ top) )
    fold)
    )) (♯ bot)) (fold))) (♯ top) )
    fold)

eq_proof : s_sub_t ≡ s_sub_t2
eq_proof = refl -- this will error
\end{verbatim}

However, in $\CoLFCircular$, a similar proof will be accepted as correct.

\begin{verbatim}
tp : type.
bot : tp.
top : tp. 
arr : tp -> tp -> tp.
mu : (tp -> tp) -> (tp -> tp) -> tp.

subtp : tp -> tp -> type.
subtpinf : tp -> tp -> cotype.
subtp/top : subtp T top.
subtp/bot : subtp bot T.
refl : subtp T T.
trans : subtp T1 T2 -> subtp T2 T3 -> subtp T1 T3.
subtp/arr : subtpinf T1 T2 -> subtp T1 T2.
unfold : {T1}{T2} 
        subtp (mu T1 T2) (arr (T1 (mu T1 T2)) (T2 (mu T1 T2))).
fold : {T1}{T2} 
        subtp (arr (T1 (mu T1 T2)) (T2 (mu T1 T2))) (mu T1 T2).

inf/arr : subtp T1 S1 -> subtp S2 T2 
        -> subtpinf (arr S1 S2) (arr T1 T2).

s : tp = mu ([x] x) ([x] x).
t : tp = mu ([x] arr x bot) ([x] top).

s_sub_t : subtp s t = 
    trans 
        (unfold ([x] x) ([x] x))
        (trans 
            (subtp/arr 
                (inf/arr
                    (trans
                        (subtp/arr
                            (inf/arr 
                                s_sub_t
                                subtp/bot))
                        (fold ([x] x) ([x] x)))
                    subtp/top))
            (fold ([x] arr x bot) ([x] top))).

s_sub_t2 : subtp s t = 
    trans 
        (unfold ([x] x) ([x] x))
        (trans 
            (subtp/arr 
                (inf/arr
                    (trans
                        (subtp/arr
                            (inf/arr 
    (trans 
        (unfold ([x] x) ([x] x))
        (trans 
            (subtp/arr 
                (inf/arr
                    (trans
                        (subtp/arr
                            (inf/arr 
                                s_sub_t2
                                subtp/bot))
                        (fold ([x] x) ([x] x)))
                    subtp/top))
            (fold ([x] arr x bot) ([x] top))))
                                subtp/bot))
                        (fold ([x] x) ([x] x)))
                    subtp/top))
            (fold ([x] arr x bot) ([x] top))).

eqsub : subtp S T -> subtp S T -> type.
eqsub/refl : eqsub M M.
eqproof : eqsub s_sub_t s_sub_t2 = eqsub/refl.
\end{verbatim}

\section{Encoding of a Polarized Subtyping System}

\label{appendix:encoding_polarized_subtyping}

We present an encoding of a variant Lakhani et al.'s polarized subtyping system \cite{Lakhani22esop} into $\CoLFCircular$. 
The system is circular. 
However, because of the awkwardness of the current LF methodology in 
encoding labelled types, e.g., $\oplus \{l : \tau_l^+\}_{l \in L}$, 
we instead fall back to the usual binary structure and write $\tau \oplus \tau'$.

\subsection{Encoding of Equirecursive Types}

The types are stratified into positive types classifying values and 
negative types classifying computations. The equirecursive nature is captured 
by a signature $\Sigma$ providing recursive definitions for type names $t^+, s^-$.
We encode the normal form of the signature which alternates between names and definitions.

\begin{center}
\begin{tabular}{lll}
    $\tau^-, \sigma^- $ & $ ::= $ & $t^+ \to s^- \mid \top \mid  s_1^- \ampersand s_2^-  \mid \upshift t^+ $ \\ 
    $\tau^+, \sigma^+ $ & $ ::= $ & $t_1^+ \otimes t_2^+ \mid \mathbf{1} \mid t_1^+ \oplus t_2^+ \mid \mathbf{0}\mid \downshift s^-$ \\
    $\Sigma $ & $ ::= $ & $\cdot \mid \Sigma, t^+ = \tau^+ \mid \Sigma, s^- = \sigma^-$ \\
\end{tabular}
\end{center}

Equirecursive types are directly encoded as recursion constants in the system, 
and the framework automatically provides equirecursive type equality checking.
Because types can be circular, both positive types and negative types are 
encoded uniformly as $\cotype$.

\begin{verbatim}
postp : cotype.
negtp : cotype.

times : postp -> postp -> postp.
one : postp.
plus : postp -> postp -> postp.
zero : postp.
downshift : negtp -> postp.

arr : postp -> negtp -> negtp.
top : negtp.
and : negtp -> negtp -> negtp.
upshift : postp -> negtp.

\end{verbatim}

The encoding relation is defined by induction on the structure of the term.
The base case of the induction is where we encounter a definitional type constant in the 
object logic (equirecursive types), and we encode the definitional type constant as a recursion variable.

\begin{longtable}{lll}
    $\enc{t^+ = \tau^+} $ & $ = $ & $  \mt{t : postp} = \enc{\tau^+} \in \Sigma$ \\
    $\enc{s^- = \sigma^-} $ & $ = $ & $ \mt{s : negtp} = \enc{\sigma^-} \in \Sigma$ \\
    $\enc{t^+ \to s^-} $ & $ = $ & $ \mt{arr}\, \enc{t^+}\, \enc{s^-}$ \\
    $\enc{\top} $ & $ = $ & $ \mt{top}$ \\
    $\enc{s_1^- \ampersand s_2^-} $ & $ = $ & $ \mt{and}\, \enc{s_1^-}\, \enc{s_2^-}$ \\
    $\enc{\upshift t^+} $ & $ = $ & $ \mt{upshift}\, \enc{t^+}$ \\
    $\enc{ s^-} $ & $ = $ & $\mt{s}$ \\
    $\enc{ t_1^+ \otimes t_2^+} $ & $ = $ & $ \mt{times}\, \enc{t_1^+}\, \enc{t_2^+}$ \\
    $\enc{\mb 1} $ & $ = $ & $ \mt{one}$ \\
    $\enc{ t_1^+ \oplus t_2^+} $ & $ = $ & $ \mt{plus}\, \enc{t_1^+}\, \enc{t_2^+}$ \\
    $\enc{\mb 0} $ & $ = $ & $ \mt{zero}$ \\
    $\enc{\downshift s^-} $ & $ = $ & $ \mt{downshift}\, \enc{s^-}$ \\
    $\enc{ t^+} $ & $ = $ & $\mt{t}$\\
\end{longtable}

\begin{theorem}[Adequacy of Type Encoding]
    There is a bijection between circular types defined in an object signature and canonical forms of the \verb#postp# or \verb$negtp$ 
    in $\CoLFCircular$.
\end{theorem}

\begin{proof}
    By induction on the structure of the object types in one direction, and by induction on the structure of terms in the reverse direction.
\end{proof}

\subsection{Encoding of Subtyping Rules}

The cyclic subtyping proof is defined via three judgments on a normal form of 
the signature, and they are $t\tpempty$, $t \tpfull$, and $t \le s$. The following shows the encoding of these judgments.

\begin{verbatim}
empty : postp -> cotype.
full : negtp -> type.
psubtp : postp -> postp -> cotype.
nsubtp : negtp -> negtp -> cotype.
\end{verbatim}

We repeat the encoding of the emptiness judgment.

\begin{mathpar}
    \inferrule{
    }{
        0 \tpempty
    }(\mathbf{0} \EMP)

    \inferrule{
        t = t_1 \oplus t_2 \in \Sigma
        \\
        t_1 \tpempty
        \\
        t_2 \tpempty
    }{
        t \tpempty
    }(\oplus \EMP)

    \inferrule{
        t = t_1 \otimes t_2 \in \Sigma
        \\
        t_1 \tpempty
    }{
        t \tpempty
    }(\otimes \EMP_1)

    \inferrule{
        t = t_1 \otimes t_2 \in \Sigma
        \\
        t_2 \tpempty
    }{
        t \tpempty
    }(\otimes\EMP_2)

\end{mathpar}


\begin{verbatim}
zero_emp : empty zero.
plus_emp : empty T1 -> empty T2 -> empty (plus T1 T2).
times_emp_1 : empty T1 -> empty (times T1 T2).
times_emp_2 : empty T2 -> empty (times T1 T2).
\end{verbatim}




The rules for fullness judgment:

\begin{mathpar}
    \inferrule{
        s = t_1 \to s_2 \in \Sigma
        \\
        t_1 \tpempty
    }{
        s \tpfull
    }(\to \FULL)

    \inferrule{
        t = \top \in \Sigma
    }{
        t \tpfull
    }(\top \FULL)

\end{mathpar}

And their encoding:

\begin{verbatim}
arr_full : empty T -> full (arr T S).
top_full:  full top.
\end{verbatim}

\begin{theorem}
    [Adequacy]
    There is a bijection between derivations of the judgment $s\tpfull$ and 
    canonical forms of the type $\mathtt{full\, \enc{s}}$.
\end{theorem}

\begin{proof}
    By induction on the structure of the circular derivation in both directions.
\end{proof}

The rules for subtyping.

\begin{mathparpagebreakable}
    \inferrule{
        t = t_1 \otimes t_2
        \\
        u = u_1 \otimes u_2
        \\
        t_1 \le u_1
        \\
        t_2 \le u_2
}{
    t \le u
}(\otimes \SUB)

    \inferrule{
        t = \mathbf{1}
        \\
        u = \mathbf{1}
}{
    t \le u
}(\mathbf{1} \SUB)

    \inferrule{
        t = t_1 \oplus t_2
        \\
        u = u_1 \oplus u_2
        \\
        t_1 \le u_1
        \\
        t_2 \le u_2
}{
    t \le u
}(\oplus \SUB)

\inferrule{
        t = \downshift s
        \\
        u = \downshift r
        \\
        s \le r
}{
    t \le u
}(\downshift \SUB)

\inferrule{
        s = t_1 \to s_2
        \\
        r = u_1 \to r_2
        \\
        u_1 \le t_1
        \\
        s_2 \le r_2
}{
    s \le r
}(\to \SUB)

\inferrule{
        s = \upshift t
        \\
        r = \upshift u
        \\
        t \le u
}{
    s \le r
}(\upshift \SUB)

\inferrule{
        s = s_1 \ampersand s_2
        \\
        r = r_1 \ampersand r_2
        \\
        s_1 \le r_1
        \\
        s_2 \le r_2
}{
    s \le r
}(\ampersand \SUB)

\inferrule{
    t \tpempty
    \\
    u = \tau^+
}{
    t \le u
}(\bot \SUB^+)

\inferrule{
    s = \upshift t
    \\
    t \tpempty
    \\
    r = \sigma^-
}{
    s \le r
}(\bot \SUB^-)

\inferrule{
    s = \sigma^-
    \\
    r \tpfull
}{
    s \le r
}(\top \SUB)
\end{mathparpagebreakable}

The above circular rules may be encoded as 

\begin{verbatim}
tensor_sub :  psubtp T1 U1 -> psubtp T2 U2 -> psubtp (times T1 T2) (times U1 U2).
unit_sub :  psubtp one one.
or_sub :  psubtp T1 U1 -> psubtp T2 U2 -> psubtp (plus T1 T2) (plus U1 U2).
downshift_sub : nsubtp S R -> psubtp (downshift S) (downshift R).
arr_sub :  psubtp U1 T1 -> nsubtp S2 R2 -> nsubtp (arr T1 S2) (arr U1 R2).
upshift_sub : psubtp T U -> nsubtp (upshift T) (upshift U).
and_sub :  nsubtp S1 R1 -> nsubtp S2 R2 -> nsubtp (and S1 S2) (and R1 R2).
bot_sub_p : empty T -> psubtp T U.
bot_sub_n : empty T -> nsubtp (upshift T) R.
top_sub : full R -> nsubtp S R.
tensor_sub :  psubtp T1 U1 -> psubtp T2 U2
                -> psubtp (times T1 T2) (times U1 U2).

unit_sub :  psubtp one one.
or_sub :  psubtp T1 U1 -> psubtp T2 U2
              -> psubtp (plus T1 T2) (plus U1 U2).
downshift_sub : nsubtp S R
                -> psubtp (downshift S) (downshift R).

arr_sub :  psubtp U1 T1 -> nsubtp S2 R2
              -> nsubtp (arr T1 S2) (arr U1 R2).
upshift_sub : psubtp T U
                -> nsubtp (upshift T) (upshift U).
and_sub :  nsubtp S1 R1 -> nsubtp S2 R2
              -> nsubtp (and S1 S2) (and R1 R2).

bot_sub_p : empty T -> psubtp T U.
bot_sub_n : empty T -> nsubtp (upshift T) R.
top_sub : full R -> nsubtp S R.
\end{verbatim}

\begin{theorem}
    [Adequacy]
    (1)
    There is a bijection between derivations of the judgment $t\le u$ and 
    canonical forms of the type $\mathtt{psubtp\, \enc{t}\, \enc{u}}$, 
    (2)
    There is a bijection between derivations of the judgment $s\le r$ and 
    canonical forms of the type $\mathtt{nsubtp\, \enc{s}\, \enc{r}}$, 
\end{theorem}
\begin{proof}
    By induction on the depth of the infinitary derivation in both directions.
\end{proof}

\subsection{Examples}
\label{subsection:examples_of_polarized_subtyping_proof}

\subsubsection{Subtyping of Lists.}

Assume we have $ \m{int}^+ \le \m{real}^+$, we want to show that $\m{intlist}^+ \le \m{reallist}^+$, 
where $\m{intlist}^+ = \mb 1 \oplus (\m{int}^+ \otimes \m{intlist}^+)$ and 
    $\m{reallist}^+ = \mb 1 \oplus (\m{real}^+ \otimes \m{reallist}^+)$.
This can be shown by the following cyclic derivation:

\[
    \infer[\oplus\SUB]{
        (\mt{il\_sub\_rl})\m{intlist}^+ \le \m{reallist}^+
    }{
        \infer[\mb 1 \SUB]{
            \mb 1 \le \mb 1
        }{
        }
        &
        \infer[\otimes\SUB]{
            \m{int}^+ \otimes \m{intlist}^+
            \le
            \m{real}^+ \otimes \m{reallist}^+
        }{
            \infer[]{
            \m{int}^+ 
            \le
            \m{real}^+
            }{}
            &
             (\mt{il\_sub\_rl})\m{intlist}^+ \le \m{reallist}^+
        }
    }
    \]

The types and the subtyping proof can be formalized in $\CoLFCircular$ as follows:

\begin{verbatim}
int : postp.
real : postp.
int_sub_real : psubtp int real.

intlist : postp = plus one (times int intlist).
reallist : postp = plus one (times real reallist).
il_sub_rl : psubtp intlist reallist = 
        or_sub (unit_sub) (tensor_sub int_sub_real il_sub_rl).
\end{verbatim}

\subsubsection{Subtyping of Computations.}
As a classic example \cite{Danielsson10mpc}, let us consider the type $\sigma^- = \downshift\sigma^- \to \sigma^-$ and 
$\tau^- =  \downshift(\downshift\tau^- \to \upshift\mb{0}) \to  \top$.
We show that $\sigma^- \le \tau^-$ by the following circular derivation:

\[
\infer[\to\SUB]{
    (\mt{eg\_s\_sub\_t})\ \sigma^- \le \tau^-
}{
    \infer[\downshift\SUB]{
        \downshift(\downshift \tau^- \le \upshift \mb{0}) \le \downshift \sigma^-
    }{
        \infer[\to\SUB]{
            \downshift \tau^- \to \upshift \mb{0} \le  \sigma^-
        }{
            \infer[\downshift\SUB]{
                \downshift \sigma^- \to \downshift \tau^-
            }{
                (\mt{eg\_s\_sub\_t})\ \sigma^- \le \tau^-
            }
            &
            \infer[\bot\SUB^-]{
                \upshift \mb{0} \le \sigma^-
            }{
                \infer[\mb{0}\EMP]{
                    \mb{0} \tpempty
                }{}
            }
        }
    }
    &
    \infer[\top \SUB]{
        \sigma^- \le \top
    }{
        \infer[\top\FULL]{
            \top \tpfull
        }{
        }
    }
}
    \]



The types and the subtyping proof can be encoded as follows:

\begin{verbatim}
eg_s : negtp = arr (downshift eg_s) (eg_s) .
eg_t : negtp = arr (downshift 
                        (arr (downshift eg_t) (upshift zero))
                    ) (top) .
eg_s_sub_t : nsubtp eg_s eg_t = 
        arr_sub 
                (downshift_sub 
                        (arr_sub
                                (downshift_sub eg_s_sub_t)
                                (bot_sub_n zero_emp)
                        )
                )
                (top_sub top_full).
\end{verbatim}

\section{Encoding Higher-Order Rational Terms and Equalities on Them}
\label{appendix:encoding_cyclic_lambda_terms}

As a meta-example, we encode the simply typed cyclic terms of $\CoLFCircular$, using an internal typing.
We encode circular terms in the object logic (i.e., $\CoLFCircular$ type theory) as  circular terms in the framework.
In this way, the equality checking of circular terms can be directly encoded as equality checking in the framework.

The syntax for internal simple typing of terms is as follows. The predicate \verb$itm$ is an \emph{intermediate} term between 
canonical terms \verb$tm$ and atomic terms \verb$atm$ that give rise to the coinductive structure of circular terms.
\begin{verbatim}
tp : type.
* : tp.
arr : tp -> tp -> tp.

tm : tp -> type.
atm : tp -> type.
itm : tp -> cotype.

lam : (atm A -> tm B) -> tm (arr A B).
base : itm A -> tm A.
at : atm A -> itm A.
app : atm (arr A B) -> tm A -> atm B.

eqtm : tm A -> tm A -> type.
eqtm/refl : eqtm M M.

\end{verbatim}
\begin{theorem}
    [Adequacy of Encoding]
    There is a bijection between simply typed terms of $\CoLFCircular$ (given by the syntax of $M$ on page $\pageref{table:syntax_of_colf_two_circular}$)
    and canonical terms of type \verb$tm$.
\end{theorem}
\begin{proof}
    By induction on the syntax in both directions. Recursion constants in the object logic correspond to 
    recursion constants of the framework.
\end{proof}

Now we present an example encoding of the term $\m{fix} = \lambda f.\, f\, (\m{fix}\, f)$ as \verb$fix$ and 
its one-step unfolding \verb$fix2$. We show that \verb$fix$ and \verb$fix2$ are equal by the proof \verb$eqfix$.

\begin{verbatim}
%% fix : (* -> *) -> * = \f. f (r f)
fix_body : atm (arr * *) -> tm * = 
    [f] base (at (app f (fix_body f))).
fix : tm (arr (arr * *) *) = 
    lam (fix_body).

fix_body2 : atm (arr * *) -> tm * = 
    [f] base (at (app f (
            base (at (app f (
            fix_body2 f
        )))
    ))).
fix2 : tm (arr (arr * *) *) = 
    lam (fix_body2).

eqfix : eqtm fix fix2 = eqtm/refl.
\end{verbatim}

As a meta-example, we consider the encoding of stream of $2$'s with one padding between each pair of $2$'s, 
with an encoding of the signature in Section \ref{subsection:finitely_padded_streams}.
We show that two different representations \verb$r$ and \verb$r'$ of the same stream are proved to be equal (\verb$eqr$)
 in the framework.

\begin{verbatim}
%% a stream of twos with single padding in between
%% r =  cocons (succ (succ zero)) (pad (next (r)))
int : tp.
zero : atm int.
succ : atm (arr int int).

pstream : tp.
padding : tp.
cocons : atm (arr int (arr padding pstream)).
next : atm (arr pstream padding).
pad : atm (arr padding padding).

r : tm pstream = base (at (app (app cocons 
    (base (at (app succ (base (at (app succ (base (at zero))))))))
   ) (base (at (app pad (base (at (app next r))))))
   )).

r' : tm pstream = base (at (app (app cocons 
  (base (at (app succ (base (at (app succ (base (at zero))))))))
  ) (base (at (app pad (base (at (app next 
   (base (at (app (app cocons 
    (base (at (app succ (base (at (app succ (base (at zero))))))))
   ) (base (at (app pad (base (at (app next r'
   )
   ))))))
   ))
  ))))))
  )).

eqr : eqtm r r' = eqtm/refl.
\end{verbatim}

\section{A Bisimulation Relation}
\label{appendix:bisimulation_example}

As an example, we could establish a bisimulation between the even/odd predicate and the conatural number predicate, which says that
every conatural number is even or odd, and every even or odd number is a conatural number.
\begin{verbatim}
conat : cotype.
cozero : conat.
cosucc : conat -> conat.

even : conat -> cotype.
odd : conat -> cotype.

ev_z : even cozero.
ev_s : odd X -> even (cosucc X).
od_s : even X -> odd (cosucc X) .

%% omega is both even and odd
omega : conat = cosucc omega.
ev_omega : even omega = ev_s (od_omega).
od_omega : odd omega = od_s (ev_omega).

%% bisimulation: every number is even or odd and
%% every odd and even number is a natural number
isconat : conat -> cotype.
isconat_z : isconat cozero.
isconat_s : isconat X -> isconat (cosucc X).

isconat_omega : isconat omega = isconat_s (isconat_omega).

bisim_ev : even X -> isconat X -> cotype.
bisim_od : odd X -> isconat X -> cotype.

bisim_ev_z : bisim_ev ev_z isconat_z.
bisim_ev_s : bisim_od D E -> bisim_ev (ev_s D) (isconat_s E).
bisim_od_s : bisim_ev D E -> bisim_od (od_s D) (isconat_s E).
\end{verbatim}


%% file: 38_cyclic_LF_inf.bbl
\begin{thebibliography}{10}
\providecommand{\url}[1]{\texttt{#1}}
\providecommand{\urlprefix}{URL }
\providecommand{\doi}[1]{https://doi.org/#1}

\bibitem{Amadio93toplas}
Amadio, R.M., Cardelli, L.: Subtyping recursive types. ACM Transactions on Programming Languages and Systems  \textbf{15}(4),  575--631 (1993)

\bibitem{Ariola97tacs}
Ariola, Z.M., Blom, S.: Cyclic lambda calculi. In: Abadi, M., Ito, T. (eds.) Theoretical Aspects of Computer Software, Third International Symposium, {TACS} '97, Sendai, Japan, September 23-26, 1997, Proceedings. Lecture Notes in Computer Science, vol.~1281, pp. 77--106. Springer, Sendai, Japan (1997). \doi{10.1007/BFb0014548}

\bibitem{Ariola97ic}
Ariola, Z.M., Klop, J.W.: Lambda calculus with explicit recursion. Information and Computation  \textbf{139}(2),  154--233 (1997). \doi{10.1006/inco.1997.2651}

\bibitem{Barendregt85book}
Barendregt, H.P.: The lambda calculus - its syntax and semantics, Studies in logic and the foundations of mathematics, vol.~103. North-Holland (1985)

\bibitem{Basold18phd}
Basold, H.: Mixed Inductive-Coinductive Reasoning Types, Programs and Logic. Ph.D. thesis, Radboud University (Apr 2018), \url{https://hdl.handle.net/2066/190323}

\bibitem{Brandt98fi}
Brandt, M., Henglein, F.: Coinductive axiomatization of recursive type equality and subtyping. Fundamenta Informaticae  \textbf{33}(4),  309--338 (1998)

\bibitem{Brotherston05tableaux}
Brotherston, J.: Cyclic proofs for first-order logic with inductive definitions. In: Beckert, B. (ed.) International Conference on Automated Reasoning with Analytic Tableaux and Related Methods (TABLEAUX 2005). pp. 78--92. Springer LNCS 3702, Koblenz, Germany (Sep 2005)

\bibitem{Brotherston11jlc}
Brotherston, J., Simpson, A.: Sequent calculi for induction and infinite descent. Journal of Logic and Computation  \textbf{21}(6),  1177--1216 (2011)

\bibitem{Cervesato96lics}
Cervesato, I., Pfenning, F.: A linear logical framework. In: Clarke, E. (ed.) Proceedings of the Eleventh Annual Symposium on Logic in Computer Science. pp. 264--275. IEEE Computer Society Press, New Brunswick, New Jersey (Jul 1996)

\bibitem{Cervesato02tr}
Cervesato, I., Pfenning, F., Walker, D., Watkins, K.: A concurrent logical framework {II}: Examples and applications. Tech. Rep. CMU-CS-02-102, Department of Computer Science, Carnegie Mellon University (2002), revised May 2003

\bibitem{Charatonik98lics}
Charatonik, W., McAllester, D.A., Niwinski, D., Podelski, A., Walukiewicz, I.: The {H}orn mu-calculus. In: Proceedings of the Thirteenth Annual IEEE Symposium on Logic in Computer Science (LICS 1998). pp. 58--69. IEEE Computer Society Press (June 1998)

\bibitem{Chen21ms}
Chen, Z.: Towards a mixed inductive and coinductive logical framework. Tech. Rep. CMU-CS-21-144, Department of Computer Science, Carnegie Mellon University (2021)

\bibitem{Courcelle83tcs}
Courcelle, B.: Fundamental properties of infinite trees. Theoretical Computer Science  \textbf{25},  95--169 (1983)

\bibitem{Danielsson10mpc}
Danielsson, N.A., Altenkirch, T.: Subtyping, declaratively. In: 10th International Conference on Mathematics of Program Construction (MPC 2010). pp. 100--118. Springer LNCS 6120, Qu{\'e}bec City, Canada (Jun 2010)

\bibitem{Djelloul08tplp}
Djelloul, K., Dao, T., Fr{\"{u}}hwirth, T.W.: Theory of finite or infinite trees revisited. Theory and Practice of Logic Programming  \textbf{8}(4),  431--489 (2008)

\bibitem{Endrullis11tcs}
Endrullis, J., Grabmayer, C., Klop, J.W., van Oostrom, V.: On equal {\(\mu\)}-terms. Theoretical Computer Science  \textbf{412}(28),  3175--3202 (2011). \doi{10.1016/j.tcs.2011.04.011}

\bibitem{Fortier13csl}
Fortier, J., Santocanale, L.: Cuts for circular proofs: Semantics and cut-elimination. In: Rocca, S.R.D. (ed.) 22nd Annual Conference on Computer Science Logic (CSL 2013). pp. 248--262. LIPIcs 23, Torino, Italy (Sep 2013)

\bibitem{Harper93jacm}
Harper, R., Honsell, F., Plotkin, G.: A framework for defining logics. Journal of the Association for Computing Machinery  \textbf{40}(1),  143--184 (Jan 1993)

\bibitem{Harper07jfp}
Harper, R., Licata, D.R.: Mechanizing metatheory in a logical framework. Journal of Functional Programming  \textbf{17}(4-5),  613--673 (2007)

\bibitem{Harper05tocl}
Harper, R., Pfenning, F.: On equivalence and canonical forms in the {LF} type theory. Transactions on Computational Logic  \textbf{6},  61--101 (Jan 2005)

\bibitem{Huet98mscs}
Huet, G.P.: Regular {B}{\"{o}}hm trees. Mathematical Structures in Computer Science  \textbf{8}(6),  671--680 (1998), \url{http://journals.cambridge.org/action/displayAbstract?aid=44783}

\bibitem{Lakhani22esop}
Lakhani, Z., Das, A., DeYoung, H., Mordido, A., Pfenning, F.: Polarized subtyping. In: Sergey, I. (ed.) Programming Languages and Systems - 31st European Symposium on Programming, {ESOP} 2022, Munich, Germany, April 2-7, 2022, Proceedings. Lecture Notes in Computer Science, vol. 13240, pp. 431--461. Springer (2022). \doi{10.1007/978-3-030-99336-8\_16}

\bibitem{Ligatti17toplas}
Ligatti, J., Blackburn, J., Nachtigal, M.: On subtyping-relation completeness, with an application to iso-recursive types. ACM Transactions on Programming Languages and Systems  \textbf{39}(4),  4:1--4:36 (Mar 2017)

\bibitem{Miller91jlc}
Miller, D.: A logic programming language with lambda-abstraction, function variables, and simple unification. Journal of Logic and Computation  \textbf{1}(4),  497--536 (1991). \doi{10.1093/logcom/1.4.497}

\bibitem{Miller05tocl}
Miller, D., Tiu, A.: A proof theory for generic judgments. ACM Transactions on Computational Logic  \textbf{6}(4),  749--783 (2005). \doi{10.1145/1094622.1094628}

\bibitem{Pfenning98pstt}
Pfenning, F., Sch{\"u}rmann, C.: Algorithms for equality and unification in the presence of notational definitions. In: Galmiche, D. (ed.) Proceedings of the CADE Workshop on Proof Search in Type-Theoretic Languages. Electronic Notes in Theoretical Computer Science (Jul 1998)

\bibitem{Pfenning98guide}
Pfenning, F., Sch{\"u}rmann, C.: Twelf User's Guide, 1.2 edn. (Sep 1998), available as Technical Report CMU-CS-98-173, Carnegie Mellon University

\bibitem{Simon06phd}
Simon, L.E.: Extending logic programming with coinduction. Ph.D. thesis, University of Texas at Dallas (2006)

\bibitem{Watkins02tr}
Watkins, K., Cervesato, I., Pfenning, F., Walker, D.: A concurrent logical framework {I}: Judgments and properties. Tech. Rep. CMU-CS-02-101, Department of Computer Science, Carnegie Mellon University (2002), revised May 2003

\end{thebibliography}
